\documentclass[final,3p]{elsarticle}
\makeatletter
\def\ps@pprintTitle{%
 \let\@oddhead\@empty
 \let\@evenhead\@empty
 \def\@oddfoot{\centerline{\thepage}}%
 \let\@evenfoot\@oddfoot}
\makeatother

\usepackage{amsmath}
\usepackage{amssymb}
\usepackage{amstext}
\usepackage{amsfonts}
\usepackage{amssymb}
\usepackage{amsthm}
\usepackage{graphicx}
\usepackage{xcolor}
\usepackage{mathtools}
\usepackage[mathcal,mathscr]{euscript}
\usepackage{microtype} 
\usepackage{tikz}
\usetikzlibrary{arrows,automata,decorations,positioning,shapes,snakes,calc}
\tikzset{>={latex},initial text={}}
\tikzstyle{accepting}=[double distance=1.7pt,outer sep=1pt+\pgflinewidth]
\tikzstyle{state}=[inner sep=1.4pt, minimum width=7mm, circle, draw,scale=0.85]
\tikzstyle{node}=[inner sep=1.3pt, minimum width=6mm, circle, draw,scale=0.80]
\usepackage{hyperref}
\usepackage{shuffle}
\usepackage{multirow}
\usepackage[shortlabels]{enumitem}
\usepackage{xcolor}
\hypersetup{linkbordercolor=blue,urlbordercolor=green, colorlinks=true, linkcolor=green!45!black, urlcolor=blue, citecolor=green!45!black}
\usepackage{witharrows}

\usepackage[capitalise,nameinlink]{cleveref}
\crefname{lem}{Lemma}{Lemmas}
\crefname{thm}{Theorem}{Theorems}
\crefname{result}{Result}{Results}
\crefname{prop}{Proposition}{Propositions}
\crefname{rem}{Remark}{Remarks}
\crefname{ex}{Example}{Examples}
\crefname{defi}{Definition}{Definitions}
\crefname{appendix}{}{}

\newtheorem{thm}{Theorem}[section]
\newtheorem{lem}[thm]{Lemma}
\newtheorem{cor}[thm]{Corollary}
\newtheorem{prop}[thm]{Proposition}
\newdefinition{defi}{Definition}[section]

\newtheorem{problem}{Problem} 

\newtheorem{result}{Result}
\newtheorem{clm}{Claim}

\theoremstyle{remark}
\newtheorem{ex}{Example}
\newtheorem{rem}{Remark}

\newcommand{\NN}{\mathbb{N}}

\newcommand{\RRpos}{\mathbb{R}_{\geq 0}}

\newcommand{\arith}[1]{\ensuremath{\mathsf{arith}_{\:\!\mathsf{C}}(#1)}} 
\newcommand{\arithf}[1]{\ensuremath{\mathsf{arith}_{\:\!\mathsf{F}}(#1)}} 
\newcommand{\boolf}[1]{\ensuremath{\mathsf{bool}_{\:\!\mathsf{F}}(#1)}} 
\newcommand{\rpn}[1]{\ensuremath{\mathsf{rpn}(#1})} 
\newcommand{\Rpn}[1]{\ensuremath{\mathsf{rpn}\big(#1}\big)} 
\newcommand{\leaf}[1]{\ensuremath{\mathsf{leaf}(#1})} 

\newcommand{\binomial}[2]{B_{#1,#2}} 
\newcommand{\binomialinf}[2]{B^+_{#1,#2}} 
\newcommand{\dyck}[1]{D_{#1}} 
\newcommand{\thresh}[2]{T_{#1,#2}} 
\newcommand{\Lxor}[1]{L^{\textup{xor}}_{#1}} 
\newcommand{\Ldiv}[2]{L^{\textup{div}}_{#1,#2}} 
\newcommand{\Ldivv}[3]{L_{#1,#2}^{#3}} 
\newcommand{\Rdiv}[3]{R_{#1,#2}^{\:\!#3}} 
\newcommand{\Leven}[2]{L^{\textup{even}}_{#1,#2}} 
\newcommand{\Levenn}[1]{L^{\textup{even}}_{#1}} 
\newcommand{\Lpar}[3]{L^{\equiv #3}_{#1,#2}} 
\newcommand{\Lpalin}[1]{L^{\textup{pal}}_{#1}} 
\newcommand{\Lperm}[1]{P_{#1}} 
\newcommand{\Lperminf}[1]{P^+_{#1}} 

\newcommand{\Aeven}[1]{\ensuremath{A_{#1}^{\textsf{even}}}} 
\newcommand{\Lnw}{\ensuremath{ L_{n,\hspace*{0.6pt}\overline{\hspace*{-0.6pt} w \hspace*{-0.6pt}}\hspace*{0.6pt}}}} 
\newcommand{\QM}{\textup{QM}_n} 

\newcommand{\eins}{\mathsf{1}}
\newcommand{\nulll}{\mathsf{0}}
\newcommand{\mplus}{\boxplus} 
\newcommand{\const}[1]{\lambda_{#1}} 
\renewcommand{\epsilon}{\varepsilon} 
\newcommand{\bin}[1]{\langle{#1}\rangle_{\textsf{2}}} 

\newcommand{\Cdot}{\:\!\raisebox{-0.5pt}{{\scalebox{1.2}{\ensuremath{\cdot}}}}\:\!}
\newcommand{\prodd}[1]{A(#1)}
\newcommand{\Hrubes}{Hrube\v{s}}
\newcommand{\isdef}{\!:=}
\newcommand{\defis}{=:\!}
\newcommand{\cupp}{\mathbin{\dot{\cup}}}
\newcommand{\equivp}{\ensuremath{\equiv_p}}
\newcommand{\minus}{\:\!\!\ensuremath{-}\:\!\!}

\newcommand{\medfrac}[2]{\scalebox{0.91}{\ensuremath{\displaystyle{\frac{\raisebox{-1.2pt}{\ensuremath{#1}}}{\raisebox{0.5pt}{\ensuremath{#2}}}}}} }
\newtagform{simple}{}{}{}
\newcommand{\clmproof}[1]{\hfill \raisebox{-0pt}{$\qedhere_{\textup{\:\!\cref*{#1}}}$}}

\newcommand{\lenv}[1]{{#1}_{\textup{le}}}
\newcommand{\henv}[1]{{#1}_{\textup{he}}}
\newcommand{\lenvv}[2]{{#1}_{#2\textup{-le}}}
\newcommand{\henvv}[2]{{#1}_{#2\textup{-he}}}

\newcommand{\reg}[1]{\textup{Reg}(#1)}

\journal{Elsevier}

\begin{document}

\begin{frontmatter}

\tnotetext[prelim]{A preliminary version of this paper appeared in Descriptional Complexity of Formal Systems \cite{CS20}.
The second author was partially supported by the DFG grant JU 3105/1-2 (German Research Foundation).}

\title{Regular expression length via arithmetic formula complexity\tnoteref{prelim}}

\author{Ehud Cseresnyes}
\ead{ehud@posteo.de} 
 
\author{Hannes Seiwert\corref{cor1}}
\ead{seiwert@em.uni-frankfurt.de}
\cortext[cor1]{Corresponding author.}

\address{Institute of Computer Science, Goethe University Frankfurt, Germany}

\begin{abstract}
We prove lower bounds on the length of regular expressions for finite languages by methods from arithmetic circuit complexity. 
First, we show a reduction: the length of a regular expression for a language $L\subseteq \{0,1\}^n$ is bounded from below by the minimum size of a monotone arithmetic formula computing a polynomial that has $L$ as its set of exponent vectors, viewing words as vectors.
This result yields lower bounds for the binomial language of all words with exactly~$k$ ones and $n\!-\!k$ zeros and for the language of all Dyck words of length~$2n$. 
We also determine the blow-up of language operations (intersection and shuffle) of regular expressions for finite languages.
Second, we adapt a lower bound method for multilinear arithmetic formulas by so-called log-product polynomials to regular expressions. 
With this method we show almost tight lower bounds for the language of all binary numbers with~$n$ bits that are divisible by~a given odd integer~$p$, for the language of all words of length~$n$ over a~$k$ letter alphabet with an even number of occurrences of each letter and for the language of all permutations of $\{1,\dots, n\}$.
\end{abstract}


\end{frontmatter}

\section{Introduction}\label{sec:intro} 
Deriving lower bounds on the length of regular expressions is a fundamental problem in formal language theory \cite{zeiger76, ellul2004, gruber2008cc, gruber2008ops1, gelade12, mousavi2017, shallit2019}. 
Particularly interesting are language families that have small finite automata but require long regular expressions, since they reveal the gap between the descriptional complexity of these models.
However, despite regular expressions being around for some decades, only few lower bounds are known so far. This is in sharp contrast to state complexity of finite automata which is understood quite well.

One has to distinguish between finite and infinite languages, as well as between alphabets of constant and growing size. 
Ehrenfeucht and Zeiger \cite{zeiger76} gave an exponential lower bound for the infinite language of all walks in a complete graph over an alphabet of size~$n^2$, 
thereby showing that the blow-up of converting a deterministic finite automaton (DFA) into a regular expression may be exponential. 
This result was generalized by Gelade and Neven~\cite{gelade12} for four-letter alphabets, and by Gruber and Holzer~\cite{gruber2008ops1} for binary alphabets using concepts of cycle rank and star height; 
in particular, Gruber and Holzer showed that the length of any regular expression is exponential in the star height of its described language.

In this paper, we focus on \emph{finite} languages. Since star height of finite languages  trivially is zero, this method is not applicable here. Instead, techniques from circuit complexity have proven useful.
In the next subsection we give an overview of the most important lower bound methods for finite languages and briefly discuss their strengths and weaknesses. Subsequently we present our own results.

\goodbreak
\subsection{Related work} \label{sec:related}

As customary, a regular expression (or just \emph{expression}) consists of the binary operations union $(+)$ and concatenation~$(\Cdot)$, the unary star operation $(^*)$, the empty language~$\emptyset$, the empty word~$\epsilon$, and letters~$a$ of an alphabet~$\Sigma$.
Every expression $R$ \emph{describes} a regular language $L(R)$ in the usual way.
Note that for finite languages the star operation is useless and can be avoided.
For an expression $R$ we measure its \emph{length} by the reverse polish length $\rpn{R}$, namely the number of nodes in its syntax tree.
For a language $L$ its \emph{expression length} $\rpn{L}$ is the length of a shortest expression describing~$L$.
A language is \emph{homogeneous} if all its words have the same length. 
For a word $w$ and a letter $a$ we denote by $|w|_a$ the number of occurrences of $a$ in~$w$.

\paragraph*{Circuits and formulas over semirings} 
A regular expression without star is nothing else than a formula (that is, a circuit whose graph is a tree) over the non-commutative \emph{free semiring} $(+,\Cdot)$.
So, lower bounds on the size of formulas (or circuits) over semirings can also be applied to expression length.
In particular, if we interpret union ($+$) as addition and concatenation ($\Cdot)$ as multiplication,  bounds for monotone arithmetic formulas (or circuits) also carry over.

For example, Jerrum and Snir \cite{jerrum1982} showed several such lower bounds on circuit size. 
By additionally taking into account the non-commutativity of the concatenation operation, they showed a lower bound of $2^n{-}2$ for the language $\Lperm{n}$ of all permutations over the alphabet $\Sigma= [n]$, see \cite[Sect.\:5.4]{jerrum1982}. (Recently, Molina Lovett and Shallit \cite{shallit2019} improved this bound to $4^n   n^{-(\log n)/4 + \Theta(1) }$ by a custom argument.)
Also \Hrubes, Wigderson and Yehudayoff \cite{hrubes10,non-commuting} developed a lower bound method for non-commutative circuits, and Filmus \cite{filmus2011} used a similar method  for lower bounds on the size of context-free grammars for finite languages. 

By using a lower bound on monotone arithmetic \emph{formula} size, \Hrubes\ and Yehudayoff \cite{yehudayoff2011} (implicitly) showed  that the language 
$L_{n,k} \isdef \{{i_1} {i_2} \cdots {i_k} : 1\leq i_1 < i_2 < \dots < i_k \leq n \}$ of all length-$k$ subwords of the word $12 \cdots n$ over the alphabet $\Sigma=[n]$ requires expressions of length $n k^{\Omega(\log k)}$, provided $k\leq n/2$.

\paragraph*{Fooling sets}

It is well known that every expression of length $s$ can be transformed into a nondeterministic finite automaton (NFA) with $s+1$ states. 
Therefore, lower bounds for NFAs carry over to expressions.
Such bounds can be shown by the fooling set method of Birget \cite{foolingset1}, rediscovered by Glaister and Shallit \cite{foolingset2}.
Given a language $L$, a \emph{fooling set} for $L$ is a set of pairs of words $(u_i,v_i)$ such that the word $u_i v_i$ lies in $L$ for all $i$, but for all $i\not = j$ one of the words $u_i v_j$ or $u_j v_i$ lies not in~$L$.
The number of states of every NFA for~$L$ is bounded from below by the size of a largest fooling set for~$L$.
\begin{ex} \label{ex:fooling}
Consider the language $\Lpalin{2n}=\{w w^\textit{reverse} : w \in \{0,1\}^{n} \} $ of all palindromes of length $2n$.
Then $F = \{(w ,w^\textit{reverse}) : w \in \{0,1\}^{n} \}$ is a fooling set for $\Lpalin{2n}$.
Hence, any NFA for $\Lpalin{2n}$ has at least $|F|=2^n$ states and any expression has length at least $2^n-1.$
\end{ex}

\paragraph*{Reduction to boolean formula complexity} \label{sec:reduction-bool}

Ellul, Krawetz, Shallit and Wang \cite{ellul2004} described a construction that transforms a given expression $R$ for a homogeneous language $L \subseteq \{0,1\}^n$ into a boolean formula of size $\rpn{R}$ for its characteristic function $f_L \!: \{0,1\}^n  \to \{0,1\}$. 
This transformation is done as follows:
Every letter $1$ (resp.~$0$) is replaced by the literal $x_i$ (resp. $\lnot x_i$), where the index $i$ corresponds to the position of the letter in the words of~$L$.
Every union is replaced by OR and every concatenation is replaced by AND.
Hence, the expression length of $L$ is bounded from below by the \emph{boolean formula complexity} of $f_L$:
$$\rpn{L} \geq \mathsf{Boolean}(f_L)$$
Together with a result by Khrapchenko \cite{khrapchenko1971} on the formula complexity of the XOR function, they derived a tight lower bound of $\Omega(n^2)$ for the XOR  language $\Lxor{n} = \left\{ w \in \left\{ 0,1 \right\}^n : |w|_1 \text{ is even} \right\}$.

\paragraph*{Reduction to monotone boolean formula complexity} \label{sec:reduction-mon}

A language $L\subseteq \{0,1\}^*$ is \emph{monotone} if it is closed under replacing any number of zeros by ones. 
By methods from communication complexity, Gruber and Johannsen \cite{gruber2008cc} developed a construction that transforms a given expression for a \emph{monotone} homogeneous language $L \subseteq \{0,1\}^n$ into a \emph{monotone} boolean formula for~$f_L$.
Thus, the expression length of $L$ is bounded from below by the \emph{monotone} boolean formula complexity of~$f_L$:
$$\rpn{L} \geq \mathsf{Monotone}\text{-}\mathsf{Boolean}(f_L)$$
With this reduction they showed that the blow-up for converting a DFA accepting a \emph{finite} language into an expression may be as large as $n^{\Theta(\log n)}$.
This result is based upon a lower bound on the monotone boolean complexity for the FORK problem shown by Grigni and Sipser \cite{grigni}.
Since any DFA (or NFA) for a finite language can be simulated by an expression of length $n^{O( \log n)}$ \cite{ellul2004,gruber2008cc}, this result is tight.
In particular, expression length of finite languages with polynomial size NFAs is at most quasi-polynomial.

\paragraph*{Strengths and limitations}
We give a brief discussion of the above methods, focusing on two desirable properties: (i) the ability to show blow-ups between finite automata and expressions, and (ii) the possibility to work over constant size alphabets.

General circuit lower bound methods may work well for languages over alphabets of growing size. 
Often, however, and in particular when working over constant size alphabets, non-commutativity must be used explicitly. For example, if we ignore non-commutativity, the permutation language $P_n$ collapses to the trivial language $\{12 \cdots n\}$ containing only a single word.
Moreover, circuit methods \emph{cannot} show nontrivial blow-ups since NFAs can be simulated by circuits (see \cite{italian2010}). 
However, \emph{formula} size methods \emph{can} (e.g., the aforementioned language $L_{n,k} = \{{i_1} \cdots {i_k} : 1\leq i_1 < \dots < i_k \leq n \}$ has small DFAs with $O(nk)$ states, but no short expressions).

The fooling set method can be applied to languages over arbitrary alphabets. It is usually simpler in its application than circuit methods and gives the same (sometimes even better) bounds, so it is often a promising first attempt. 
However, since lower bounds shown with fooling sets hold also for NFAs, clearly this method cannot show blow-ups either.  

The two ``boolean methods'' in contrast \emph{can} show blow-ups and also work for constant size (even binary) alphabets. 
On the downside, they rely on lower bounds for boolean (or monotone boolean) complexity which are rare and whose proofs are usually quite involved. Also, bounds obtained by the boolean methods can be rather loose for many languages (see \cref{sec:limits} and \cref{rem:arithvsbool,rem:threshold}) and are restricted to languages over the alphabet $\{0,1\}$.

\paragraph*{Our contribution}

Our goal in this paper is to improve the existing methods with regard to strength as well as to simplicity of application.
We present two new methods that naturally refine the boolean methods resp. the formula size method. Both methods have our two desired properties (i) and (ii), that is, they are able to show blow-ups and work over constant size alphabets. 
As demonstration of these methods, we also show several explicit lower bounds.

\section{Results} \label{sec:results}

\paragraph*{Reduction to monotone arithmetic formula complexity}
A \emph{monotone arithmetic formula} is a rooted tree with leaves holding variables $x_1, \dots, x_n$ and inner nodes (\emph{gates}) performing multiplication~$(\times)$ or addition~$(+)$ operations. 
Any such formula computes a polynomial $f(x_1, \dots, x_n)= \sum_{a \in A} \const{a} \prod_{i=1}^n x_i^{a_i}$ over the nonnegative reals in a natural manner, where  $A\subseteq \NN^n$ is its set of \emph{exponent vectors} and $\const{a}>0$ are positive  coefficients.
For a set $A \subseteq \NN^n$ let $\arithf{A}$ be the size of a smallest monotone arithmetic formula that computes a polynomial whose set of exponent vectors is $A$. We identify words $w=w_1 \cdots w_{n}$ with  vectors $(w_1, \dots, w_n)$.

Our first method reduces expression length to monotone arithmetic formula complexity.
\begin{result}[Arithmetic bound, \cref{thm:arith}]\label{res:arith}
Let $L \subseteq \{0,1\}^n$ be a homogeneous language. Then  $$\rpn{L} \geq \arithf{L}\, .$$
\end{result}
Informally stated, the following hierarchy for the different formula complexities holds:
$$ \mathsf{Monotone}\text{-}\mathsf{Arithmetic} \geq \mathsf{Monotone}\text{-}\mathsf{Boolean} \geq \mathsf{Boolean}$$
The gaps between these complexities can be exponentially large (see \cref{rem:arithvsbool}). So, our arithmetic bound covers both aforementioned boolean methods, and can be exponentially stronger.
Another advantage is that bounds on arithmetic complexity can be proven much more easily than on boolean complexity. 
In particular, there are already many strong bounds known.
In contrast to the ``monotone boolean method'' of Gruber and Johannsen \cite{gruber2008cc}, the arithmetic bound is not restricted to monotone languages.
For a survey on boolean complexity resp. arithmetic complexity, including the non-monotone case, see \cite{jukna2012} resp. \cite{amir10:survey,survey:github}.

\paragraph*{The log-product bound}
A general flaw of \cref{res:arith} is that the non-commutativity of the concatenation operation cannot be fully utilized since arithmetic operations ($+$ and $\times$) are commutative. 
Further, it is restricted to languages over the alphabet $\Sigma=\{0,1\}$.
To cope with these issues, we adapt a lower bound method from Shilpka and Yehudayoff \cite{amir10:survey} resp. \Hrubes\ and Yehudayoff \cite{yehudayoff2011} for multilinear arithmetic formula size \emph{directly} to expression length. 
(Essentially, we include non-commutativity in their lower bound method for formula size, as similarly done in \cite{hrubes10,non-commuting} for a circuit size method.) 

An expression $R$ is \emph{homogeneous} if it describes a  homogeneous language (all words have the same length) and its \emph{degree} $\deg R $ is the length of its described words.
A homogeneous expression $B$ is \emph{log-product} if it is either a letter or a concatenation of two homogeneous expressions $B_1, B_2$ such that $\deg B_1 \geq \deg B_2$ and $B_1$ is log-product itself. 
We show that every homogeneous expression $R$ can be written as union $B_1 + \dots + B_{\ell}$  of $\ell \leq \rpn{R}$ log-product expressions $B_i$. This yields our second method.
\begin{result}[Log-product bound, \cref{thm:logprod}] \label{res:logprod} 
Let $L$ be a homogeneous language and $h\in \RRpos$. 
If $|L(B)| \leq h$ holds for every log-product expression $B$ with $L(B) \subseteq  L$, 
then any expression for $L$ has length at least $$\rpn{L} \geq |L| /h \,.$$
\end{result}
So, in order to show a lower bound on $\rpn{L}$, we only have to upper bound the number of words in $L(B)$ for any log-product expression $B$ with $L(B) \subseteq L$. 
Every log-product expression can be written as factorization $B=F_1 F_2 \cdots F_m$ of $m \geq \log (\deg B)$ nontrivial factors $F_i$ (see \cref{sec:factor} for details).
Since every word described by $B$ lies in $L$, we can derive properties of languages described by the factors $F_i$.
For example, if $L\subseteq \{0,1\}^n$ is a \emph{uniform} language, that is, every word has the same number of ones, then every language $L(F_i)$ described by a factor $F_i$ must be uniform as well. 
These properties then allow us to upper bound the number of words described by each factor, and hence, by $B$.

\paragraph*{Lower bounds for explicit languages}
We apply our two methods (\cref{res:arith,res:logprod}) to several explicit language families. Namely, we prove lower bounds for
\begin{itemize}[itemsep=-1.0pt]
\item
the binomial language $\binomial{n}{k} =\{ w\in \{0,1\}^n :|w|_1=k \}$,
\item 
the threshold language $ \thresh{n}{k}=\{ w\in \{0,1\}^n :|w|_1 \geq k \}$,
\item  
the language $\dyck{2n}$ of all length $2n$ Dyck words, 
\item
the divisibility language $\Ldiv{n}{p} $ of all binary numbers with $n$ bits that are divisible by $p$, 
\item
the parity language $\Leven{n}{k}= \{w \in \Sigma^n: |w|_i \text{ is even for all }i \in \Sigma \}$ over $\Sigma=\{1,\dots,k\}$, and
\item the language $P_n$ of all permutations of $\Sigma=\{1,\dots, n\}$.
\end{itemize}
The results are summarized in \cref{table:results}, for comparison we also list upper bounds on expression length and DFA size. 
These DFAs are more or less trivial, and the upper bounds follow from a simple conversion of the DFAs (see \cref{prop:dfa2regex}). 
All these languages reveal a large gap between DFA size and expression length.
The first three bounds follow from the arithmetic bound (\cref{res:arith}), the other three are derived with the log-product bound (\cref{res:logprod}).

\begin{table}[t]
\setlength{\tabcolsep}{2mm}
\def\arraystretch{1.25}
\begin{center}
\begin{tabular}{c|c|c|c}
language & DFA & \textsf{rpn}, upper bound & \textsf{rpn}, lower bound    \\
\hline 
$\phantom{\Big\{ } \binomial{n}{k} \phantom{\Big\{ }$ & $O(nk)$ & $O(nk \cdot n^{\log (k+1) })$ & $n  k^{\Omega(\log k) }~$ [\cref{cor:binomial}] \\
$\thresh{n}{k}$ & $O(nk)$ & $O(nk \cdot n^{\log (k+1) })$ & $n  k^{\Omega(\log k) }~$  [\cref{cor:threshold}]   \\
$\dyck{2n}$ & $O(n^2)$ & $O(n^{\log n + 3} )$ & $n^{\Omega(\log n) }~$ [\cref{cor:dyck}] \\
$\Ldiv{n}{p}$ & $O( n\:\!p )$ & $O(n\:\!p\cdot p^{\log (n/\log p)})$ & $\Omega( n^{-1} p^{\log (n / \log p) -2 })~$ [\cref{thm:div}]  \\
$\Leven{n}{k}$  & $O(2^k \:\!n)$   & $O(k \:\!2^k  n^{k})$ &  
$   n^{(k-2)} \, k^{ -\Theta(k)}$ [\cref{thm:even}]  \\
$\Lperm{n}$ & $O(2^n)$ & $4^n n^{- (\log n)/4+O(1)}$ & $\Omega(4^n n^{- (3+\log n)/4})~$ [\cref{thm:perm}]
\end{tabular}
\end{center}
\vspace*{-2mm}
\caption{Upper and lower bounds on expression length and DFA sizes for several languages.}\label{table:results}
\end{table}

In particular, we answer a question from Ellul et al. \cite{ellul2004} asking for the length of optimal expressions for the binomial language $\binomial{n}{k} =\{ w\in \{0,1\}^n :|w|_1=k \}$, namely for $k= n^{\Theta(1)}$, length $n^{\Theta(\log n)}$ is optimal. 
In \cref{cor:uniform} we show  a superpolynomial lower bound for all \emph{uniform} languages that contain sufficiently many words. 
This yields our lower bound for the language $\dyck{2n}$ of all length $2n$ Dyck words. 
The bound for the threshold language $\thresh{n}{k}$ is an easy consequence of our bound for the binomial language.

Ellul et al. \cite{ellul2004} also asked for lower bounds for the language of all (arbitrary length) binary numbers that are divisible by a number $p$. So, our lower bound for the divisibility language $\Ldiv{n}{p}$ answers the ``finite variant'' of this question.
Our lower bound for the parity language  $\Leven{n}{k}$ naturally generalizes the lower bound from Ellul et al. \cite[Thm.\:23]{ellul2004} for the XOR language $\Lxor{n}= \{ w \in \{0,1\}^n : |w|_1 \text{ is even} \}$ to non-binary alphabets.
The bound for the permutation language $\Lperm{n}$ was already shown by Molina Lovett and Shallit \cite{shallit2019}, however, we give an alternative, shorter proof.

\paragraph*{Blow-up of language operations and DFA conversion}
A classical question is: by how much can expression length increase when performing operations like intersection, shuffle (also known as interleaving) or complementation? For infinite languages this problem was solved by Gelade and Neven \cite{gelade12} and Gruber and Holzer \cite{gruber2008ops1,gruber2009ops2}: The blow-up is exponential for intersection and shuffle, and even double-exponential for complementation. 
In contrast, for \emph{finite} languages the blow-up for intersection and shuffle is at most $n^{O(\log n)}$. 
We show that for some languages this is inevitable: There are finite languages $L_1,L_2$ with  expressions of length $O(n)$ such that any expression describing their intersection $L_1 \cap L_2$ or their shuffle $L_1 \shuffle L_2$ requires length $n^{\Omega(\log n)}$.
We prove this as a consequence of our lower bound for the binomial language.

Gruber and Johannsen \cite{gruber2008cc} showed that the blow-up of converting a DFA for a finite language into an expression can be as large as $n^{(\log n)/192}$. 
Ellul et al.\ \cite{ellul2004} showed an upper bound of $n^{\log (n) + O(1)}$ for such conversions (see \cref{prop:dfa2regex-old}), so the lower bound is tight apart from the constant factor in the exponent.
We will improve this factor to $1/4 - o(1)$.
An overview of the blow-ups is given in \cref{table:blowup}.

\newcommand{\RE}{\text{RE}}
\begin{table}[t]
\crefname{cor}{Cor.}{Cor.}
\crefname{prop}{Prop.}{Props.}
\crefname{thm}{Thm.}{Thms.}
\begin{center}
\begin{tabular}{c|c|c}
conversion & upper bound &  lower bound    \\ \hline 
$\RE \cap \RE,\; \RE \shuffle \RE \to \RE$ &  $n^{O(\log n) } \phantom{\Big\{ }$ & $n ^{\Omega(\log n) }~$ [\cref{thm:ops}]  \\
DFA,\;NFA $\to$ RE  & $n^{\log(n) +O(1)}$ \cite{ellul2004,gruber2008cc} & $n^{(\log n)/4 - \Theta( \log \log n)}$ [\cref{cor:conversion}]  \\
\end{tabular}
\end{center}
\vspace*{-2mm}
\caption{Blow-ups of operations and conversion for finite languages (assuming an alphabet of size  at most $n^{O(1)}$).} \label{table:blowup}
\end{table}

\paragraph*{Non-homogeneous and infinite languages}
All results mentioned so far are obtained for \emph{homogeneous} languages (i.e., all words have the same length). 
For non-homogeneous and even infinite languages we show that their expression length is bounded from below by the expression length of their lower and higher \emph{envelopes}, that is, their sublanguages of all shortest resp.\ all longest words, see \cref{lem:envelopes}.

By this lemma our lower bounds for the binomial language $\binomial{n}{k}$ and the permutation language~$\Lperm{n}$ carry over to their infinite variants
\begin{itemize}[itemsep=0pt]
\item 
$\binomialinf{n}{k}  = \{w \in \{0,1\}^* : |w|_1 \geq k, |w|_0 \geq n-k \}$ and
\item 
$\Lperminf{n} = \{w \in [n]^* : |w|_a \geq 1 ~\text{for all } a \in [n] \hspace*{1pt} \}$.
\end{itemize}
This is particularly interesting since the previous known lower bound methods for infinite languages (relying on star height) do not work for these languages.

\subsection{Organization} \label{sec:orga}

In the next section (\cref{sec:preliminaries}) we recall basic concepts of languages and expressions, show how to transform automata into  expressions and introduce monotone arithmetic formulas.
In \cref{sec:arith} we prove our arithmetic bound (\cref{res:arith}) and discuss some properties.
We then apply the arithmetic bound to uniform languages (\cref{sec:uniform}), investigate the blow-up of language operations (intersection and shuffle) for finite languages (\cref{sec:ops}) and address limitations (\cref{sec:limits}).
In \cref{sec:direct} we prove our log-product bound (\cref{res:logprod}) and show some useful factorization properties of log-product expressions (\cref{sec:factor}).
Applications of the log-product bound are demonstrated in \cref{sec:app-bal}, 
namely we show lower bounds for the divisibility language $\Ldiv{n}{p}$ (\cref{sec:div}),
the parity language $\Leven{n}{k}$ (\cref{sec:par}) and the permutation language~$\Lperm{n}$ (\cref{sec:perm}).
In \cref{sec:envelope} we treat non-homogeneous and infinite languages with the help of envelopes.
Finally, we summarize our results and sketch open problems in \cref{sec:conc}.

In order to provide some background regarding the connection between expressions for finite languages and circuit complexity in general, we give an introduction to semirings and circuits in \cref{sec:semirings}. Since these concepts are not actually needed in our proofs, we moved them into the appendix.

\section{Preliminaries} \label{sec:preliminaries}

Throughout, let $\RRpos$ be the set of all nonnegative reals, $\NN=\{0,1,2,\dots\}$ the set of all nonnegative integers and $[n]=\{1,2,\dots, n\}$ the set of the first $n$ positive integers.
The logarithm of base $2$ is denoted by $\log(\cdot)$ and the natural logarithm by $\ln (\cdot)$.
We assume the reader to be familiar with basic concepts of regular languages and only recall some central aspects. An introduction can be found, for example, in \cite{hopcroft2001}.

\subsection{Languages and regular expressions}

An \emph{alphabet}~$\Sigma$ is a finite nonempty set of \emph{letters}, a \emph{word}~$w=w_1 w_2 \cdots w_n$ over~$\Sigma$ is an element of~$\Sigma^n$. As customary, $\Sigma^* = \cup_{i=0}^\infty \Sigma^i$ is the set of all words over~$\Sigma$, including the empty word~$\epsilon$.
A \emph{language} over~$\Sigma$ is a set~$L \subseteq \Sigma^*$.
For a word~$w$ its \emph{length}~$|w|$ is the total number of its letters, and for a  given letter~$a \in \Sigma$ we denote by~$|w|_a$ its number of occurrences in~$w$.
To avoid pathological situations, we assume throughout that our languages satisfy~$L \neq \{\epsilon\}, \emptyset$.

Regular expressions (or just \emph{expressions}) over an alphabet $\Sigma$ are defined recursively as follows.
\begin{itemize}[itemsep=-1pt,topsep=5pt]
    \item The symbols~$\emptyset$ and~$\epsilon$ are expressions, as well as all letters~$a \in \Sigma$.
    \item If~$R$ and~$R'$ are expressions, then so are~$R^*$,~$(R \cdot R')$ and~$(R + R')$.
\end{itemize}
Every expression~$R$ \emph{describes} a regular language~$L(R)$ as follows.
\begin{itemize}[itemsep=-1pt,topsep=5pt]
    \item $L(\emptyset)=\emptyset$, $ L(\epsilon)= \{\epsilon\}$ and $L(a)=\{a\}$ for all $a \in \Sigma$.
    \item $L(R^*)= (L(R))^*$, $ L(R\cdot R') =L(R) \cdot L(R')$ and $L(R + R') = L(R) \cup L(R')$.
\end{itemize}
As customary, we abbreviate concatenations~$R \cdot R'$  by~$RR'$ and omit parenthesis where possible.
Two expressions $R$ and $R'$ are \emph{equivalent}, denoted by~$R \equiv R'$, if they describe the same language.
With a slight abuse of notation, we will sometimes identify expressions with their languages, e.g., $(a+b)$ stands for $\{a,b\}$.
We assume throughout and w.l.o.g that no expression contains the symbol $\emptyset$ and that no expression for a finite language contains a star.
Furthermore, we identify every expression with its \emph{syntax tree}.
For a node~$u$ in the syntax tree of an expression~$R$, the subtree of~$u$ is a \emph{subexpression} of~$R$, denoted by $R_u$.
There are several measures for the \emph{length} of an expression. In this paper, we use the \emph{reverse polish length}~$\rpn{R}$ which is the number of nodes in the syntax tree of~$R$; for an overview of other length measures see \cite{ellul2004}.
For a regular language~$L$ its \emph{expression length} $\rpn{L}$ is the length of a shortest expression that describes~$L$. 

We will mainly deal with homogeneous languages.
A language $L \subseteq \Sigma^*$ is \emph{homogeneous} if all words have the same length, i.e., $L \subseteq \Sigma^n$ holds for some~$n \in \NN$. 
An expression $R$ is \emph{homogeneous} if $R$ describes a homogeneous language.
Since we do not allow the symbol $\emptyset$, every subexpression of a homogeneous expression is homogeneous itself.
The \emph{degree} of $R$, denoted by $\deg R$, is the length of its described words, and the degree of a node $u$ in the syntax tree of $R$ is the degree of its subexpression $R_u$.
For a language $L \subseteq \Sigma^*$ we call the homogeneous language $L \cap \Sigma^n$ its \emph{$n$-slice}. 

\goodbreak
\subsection{Conversions of automata} \label{sec:conversions} 

To bring our lower bounds in context, we compare them with upper bounds as well as with automata for the respective languages.  
For infinite languages, the conversion of an NFA into an expression can cause an exponential blow-up \cite{zeiger76,gelade12,gruber2008ops1}. For finite languages, however, this blow-up is at most quasi-polynomial, as the following proposition tells. We refer to \cite{gruber2015survey} for a comprehensive survey on conversions. 
All results in this subsection given for NFAs hold analogously for DFAs.

\begin{prop}[Conversion for finite languages, {\cite[Cor.\ 22]{ellul2004}} or {\cite[Cor.\ 13]{gruber2008cc}}] \label{prop:dfa2regex-old}
Let $L$ be a finite language over a $k$-letter alphabet accepted by an NFA with $s$ states. Then $L$ can be described by regular expressions of length $$\rpn{L} \leq O \big( k s^3 s^{\log s} \big) \,.$$
\end{prop}
This conversion can be refined for the subclass of ``layered'' automata, and in particular for  homogeneous languages.
Call an NFA for a finite language \emph{layered} if there is a partition  $Q_0 \cup Q_1 \cup  \dots \cup Q_n$ of its states such that every word of length~$j$ leads to a state $q\in Q_j$ when given into the automaton. 
We call $Q_j$ the $j$-th \emph{layer}. 
Note that in a layered automaton only transitions between neighboring layers $Q_j$ and $Q_{j+1}$ are allowed. 
For a layered automaton define its \emph{width} as $\omega= \max_{0 \leq j \leq n} |Q_j|$ and its \emph{length} as the length $n$ of its longest path.

Every NFA\footnote{We assume that there are no ``useless'' states from which no path to a final state exists.} for a \emph{homogeneous} language is layered since every path leading from its initial state to a final state has length exactly~$n$; and in this case one final state always suffices.
In particular, given an NFA~$A$ with~$s$ states for a language $L$, we obtain a layered NFA for its $n$-slice $L\cap \Sigma^n$ by constructing the product automaton of $A$ and the trivial DFA for $\Sigma^n$; we call it the \emph{$n$-slice NFA} of $A$.\footnote{We implicitly collapse equivalent states; in particular, any slice automaton has only one initial and one final state.} This automaton has at most~$sn$ states, width~$s$ and length~$n$.

The following proposition tells us how to transform a layered automaton into an expression.\footnote{In fact, a layered automaton is nothing else than a layered \emph{branching program} over the \emph{free semiring} $(+, \Cdot)$ (see \cref{sec:semirings}) and our conversion follows the standard simulation of branching programs by formulas.}
\begin{prop}[Conversion of layered automata] \label{prop:dfa2regex}
Let $L$ be a finite language over a $k$-letter alphabet accepted by a layered NFA of width $\omega$ and length $n$ with $f$ final states. Then $L$ can be described by regular expressions of length
$$\rpn{L} \leq  O\big(  fk n  \omega \cdot \omega^{\log n} \big)\,.$$
\end{prop}

Note the difference in the bounds given by the two propositions above: The number of states of a layered automaton of width $\omega$ and length $n$ is at least $\omega+n$ and can be as large as roughly $\omega n$. 
So, for example, if the width~$\omega$ is constant, then \cref{prop:dfa2regex} yields a polynomial bound $k n^{O(1)}$ while \cref{prop:dfa2regex-old} yields only a superpolynomial bound $k n^{O( \log n)}$. 
\begin{proof}
We use a standard idea which can be found, for example, in \cite{zeiger76} or \cite{ellul2004}.
Let $A$ be an NFA for $L$, let $Q=Q_0 \cup Q_1 \cup \dots \cup Q_n$ be its set of states, $i \in Q_0$ be its initial state and $F\subseteq Q $ be its set of final states.
We recursively construct an expression $R^{x \to y}_d$ that describes all words of length $d$ that lead from state $x$ to state $y$. For $\ell \in [n]$ and $ x \in Q_\ell$ define
\begin{align}
R^{x\to y}_d= \!\! \sum_{q \in Q_{\ell+\lfloor d/2 \rfloor}}  R^{x\to q}_{\lfloor d/2 \rfloor } \cdot R^{q \to y}_{\lceil d/2\rceil } \label{eq:rec}
\end{align} for $d\geq 2$, while $R^{x\to y}_1$ is given directly by the transitions of the automaton.
Finally, the expression $R  \isdef  \sum_{j=0}^n \sum_{q\in F \cap Q_{j}} R^{i \to q}_{j}$ describes the language $L$.

For an expression $R'$ let $\leaf{R'}$ be the number of leaves in the syntax tree of $R'$, and let $T(d) = \max_{x,y} \big\{ \leaf{ R^{x\to y}_d } \big\}$. Then \cref{eq:rec} gives the recursion 
$$T(d) \leq \omega \cdot \big(\;\!T(\lfloor d/2\rfloor ) +T(\lceil d/2 \rceil) \;\! \big)$$ for $d\geq 2$, and $T(1) \leq k$. 
Solving this recursion leads to $T(d) \leq  O( k d\omega \cdot \omega^{\log d} )$ (see \cite[Lem.\ 21]{ellul2004}).
Clearly, $\rpn{R} = O(\leaf{R})$ holds and we get $\rpn{R} \leq  O( f \cdot T(n)) =  O( fk n \omega \cdot \omega^{\log n} )$.
\end{proof}

\begin{ex}[Divisibility language] \label{ex:div}
Let~$p$ be an odd integer and consider the language $L^{\textsf{div}}_p \subseteq \{0,1\}^*$ of all binary numbers that are divisible by~$p$.
This language can be accepted by a DFA with states~$Q=\{0,1,\dots, p-1\}$, initial and final state~$0$, and transitions $\delta(q,a) = (2q+a) \bmod p$ for all~$q\in Q$ and $a \in \{0,1\}$.
The $n$-slice of this language is $\Ldiv{n}{p}  =  L \cap \{0,1\}^n$ and the $n$-slice DFA has width $p$ and length $n$.
Thus, \cref{prop:dfa2regex} implies that $\Ldiv{n}{p}$ has expressions of length~$O(np \cdot p^{\log n})$. In \cref{sec:div} we will improve this upper bound to $O(np \cdot p^{\log (n/ \log p)})$
 and give an almost matching lower bound.
\end{ex}

\subsection{Monotone arithmetic formulas} \label{sec:formulas}

In this subsection we briefly treat monotone arithmetic formulas. 
A more detailed exposition of monotone arithmetic \emph{circuits} is given in \cref{sec:semirings}.

Given a number $n \in \NN$, a monotone arithmetic formula of the variables $x_1, \dots, x_n$ is a rooted tree with leaves holding either one of the variables $x_i$ or a constant $ c \in \RRpos$. 
Every inner node (a \emph{gate}) performs one of the operations addition ($+$) or multiplication ($\times$).
The \emph{size} of a formula is the number of its nodes.
Every formula computes a polynomial 
$$f(x_1, \dots, x_n)= \sum_{a \in A} \const{a} \prod_{i=1}^n x_i^{a_i}
$$ over $\RRpos$ in a natural manner, where  $A\subseteq \NN^n$ is its (finite) set of \emph{exponent vectors} and $\const{a} \in \RRpos$ are positive coefficients.
We say that the formula \emph{produces} the set $A$.

For a given set $A\subseteq \NN^n$ denote by $\arithf{A}$ the size of a smallest formula that produces $A$, i.e., that computes a polynomial whose set of exponent vectors is $A$.
Note that we do \emph{not} require the  formula to compute a polynomial with specific \emph{coefficients}, we are only interested in its \emph{set of monomials}. 
This is in agreement with almost all lower bounds shown for monotone arithmetic circuit (or formula) complexity (see, e.g., \cite[Rem.\:1]{jukna2016} or \cite{VPvsVNP}).
We identify every vector~$(w_1, w_2, \dots, w_n) \in \NN^n$ with the word $w_1 w_2 \cdots w_n$. 
Thus, the produced set $A \subseteq \NN^n$ can be interpreted as a homogeneous language over the alphabet $\Sigma= \{0,1, \dots, k\}$, with~$k$ being the largest entry of a vector in~$A$. In particular, if $A \subseteq \{0,1\}^n$ (i.e., $A$ is produced by a \emph{multilinear} formula), then $A$ is a language over $\Sigma= \{0,1\}$.

\section{Reducing expression length to monotone arithmetic formula size} \label{sec:arith}

Let $n \geq 1$, $L\subseteq \{0,1\}^n$ be a homogeneous language and $R$ be a homogeneous expression describing~$L$.
In this section, we assume w.l.o.g.\ that $R$ does not contain the symbol $\epsilon$ (in addition to not containing the symbol~$\emptyset$); hence, every leaf in $R$ is a letter.
Ellul et al. \cite[Lem.\:24]{ellul2004} transformed $R$ into a boolean formula for the function $f_L\!: \{0,1\}^n \to \{0,1\}$ such that $f_L(x)=1$ iff $x \in L$.
Namely, they assigned a unique \emph{position} $i \in [n]$ to each leaf of $R$, such that its letter occurs as the $i$-th letter in all words in $L$. For example, in the expression $(a{+}b)(cd{+}ce)$ the positions of the leaves holding the letters $a,b,c,d,e$ are $1,1,2,3,3$, respectively.
The transformation is as follows: Replace each union by OR, each concatenation by AND, and each leaf at position~$i$ holding the letter~$1$ (resp. $0$) by the literal $x_i$ (resp. $\lnot x_i$).
We present a similar transformation of~$R$ into a monotone \emph{arithmetic} formula that produces the set $L$. 

The \emph{arithmetic version of $R$} is the monotone arithmetic formula $\Phi_R$ of the variables $x_1, \dots, x_n$ that is constructed as follows: 
Replace each union node $(+)$ of $R$ by an arithmetic addition gate $(+)$ and each concatenation node $(\Cdot)$ by an arithmetic multiplication gate $(\times)$. 
Replace each leaf holding the letter~$0$ by the constant $1$ and replace each leaf holding the letter $1$ at position~$i$ by the variable $x_i$.
Note that in any case a leaf holding a letter $\sigma \in \{0,1\}$ at position $i$ is replaced by $x_i^\sigma$.
For example, the arithmetic version of the expression $000+011+100$ is the formula $1+ x_2 x_3 + x_1$ which produces the set $\{(0,0,0),(0,1,1),(1,0,0)\}$.

\begin{lem} \label{lem:trans} 
Let $R$ be a homogeneous expression with $L(R) \subseteq \{0,1\}^n$ and $\Phi_R$ be its arithmetic version.
Then  $\Phi_R$ has size at most $\rpn{R}$ and produces the set $L(R)$. 
\end{lem}
That is, the arithmetic formula $\Phi_R$ computes a polynomial $\sum_{a \in A} \const{a} \prod_{i=1}^n x_i^{a_i}$ with $A=L(R)$ and some coefficients $\const{a}> 0$.
\begin{proof}
The claim for size is trivial. To show the claim for the produced set, we first bring both $R$ and $\Phi_R$ in ``sum-product normal form'', that is, we move all union or addition nodes to the top by iteratively applying the distributive law: replace $u \cdot (v + w)$ by $ u  v + uw$ and $ (u + v) \cdot w $ by $u w + v w$.
By this procedure, the language $L(R)$ described by $R$ and both the computed polynomial and the produced set of $\Phi_R$ do not change. Hence, it suffices to show the claim for these modified versions of $R$ and~$\Phi_R$.

For the rest of the proof assume that all union nodes in $R$ and all addition gates in $\Phi_R$ are at the top.
Call a subexpression of $R$ describing a single word of length $n$ a \emph{singleton}. 
Hence, the expression $R$ is a union over all words $w \in L(R)$ (possibly with repetitions) of singletons, each describing one of the words $w$.
By construction of $\Phi_R$, every singleton describing $w=w_1  \cdots w_n$ is replaced by a subformula that computes the monomial $x^w \isdef x_1^{w_1} \cdots x_n^{w_n}$.
The entire formula $\Phi_R$ is a sum over all monomials $x^w$ for $w \in L(R)$, that is, it computes the polynomial $f(x) = \sum_{w \in L(R)} \const{w} x^w $, where the constant $\const{w}>0$ denotes the number of singletons in $R$ that describe the same word $w$.
Hence, $\Phi_R$ produces $L(R)$, as desired.
\end{proof}

Recall that $\arithf{L}$ is the size of a smallest monotone arithmetic formula that produces $L$.
\cref{lem:trans} directly yields the following theorem.
\begin{thm}[Arithmetic bound] \label{thm:arith}
Let $L \subseteq \{0,1\}^n$ be a homogeneous language. Then any regular expression describing $L$ has length at least
$$\rpn{L} \geq \arithf{L}\,.$$
\end{thm}
We make some remarks before turning to the applications.

\begin{rem} \label{rem:note}
Note that lower bounds on the size of a monotone arithmetic formula yield bounds for \emph{two different} related languages.
As a simple example, take the monotone arithmetic formula $f = 1 + x_2 x_3 + x_1$.
Then lower bounds on the formula size of $f$ carry over to the expression length of two different languages $L$ and $L'$, namely
\begin{itemize}[itemsep=0pt]
\item the language $L = \{ \epsilon, x_2 x_3 , x_1 \}$ of all \emph{monomials} of $f$  over the alphabet $\Sigma=\{ x_1, x_2, x_3 \}$, and 
\item the language $L' = \{000, 011, 100 \}$ of all \emph{exponent vectors} of $f$ over the alphabet $\Sigma=\{0,1\}$.
\end{itemize}
The first claim is long known, see for example \cite{jerrum1982}, and we already mentioned it in \cref{sec:related}.
The second claim is our arithmetic bound \cref{thm:arith}.
\end{rem}

\begin{rem}[Invariance under permutations] \label{rem:perm}
Since arithmetic operations $+$ and~$\times$ are commutative, the order of the variables does not matter, i.e., reordering the variables $x_1, \dots, x_n$ in a polynomial does not change its arithmetic complexity.
For a language $L \subseteq \{0,1\}^n$ and a permutation $\sigma \!:[n] \to [n]$ define $\sigma(L)  \isdef  \{w_{\sigma(1)}  \cdots w_{\sigma(n)} : w_1 \cdots w_n \in L\}$.
Then $\arithf{L} = \arithf{ \sigma(L) }$ holds, and all bounds shown with \cref{thm:arith} for $L$ also hold for the permuted language $\sigma(L)$. 

On the one hand, this is a nice feature; on the other, it limits the possibilities for application.
In \cref{sec:limits} we address this issue in detail.
\end{rem}

\begin{rem}[Arithmetic vs.\ boolean complexity] \label{rem:arithvsbool}

Every language $L \subseteq \{0,1\}^n$ also defines a monotone \emph{boolean} function $f^{\textsf{mon}}_L : \{0,1\}^n \to \{0,1\}$ with $f^{\textsf{mon}}_L(x) = 1 $ iff there is a word $y \in L$ such that $x\geq y$ holds componentwise.
Let $\boolf{L}$ be the size of a smallest monotone boolean formula that computes the function $f^{\textsf{mon}}_L$. 
Gruber and Johannsen \cite{gruber2008cc} showed that the expression length of a (monotone) language $L \subseteq \{0,1\}^n$ is bounded from below by the monotone boolean formula size of $f^{\textsf{mon}}_L$, that is, they showed $\rpn{L} \geq \boolf{L}$ (see \cref{sec:reduction-mon}).

\Cref{thm:arith} strictly improves on this bound.
It is well known that
$$\boolf{L} \leq \arithf{L} $$ holds for every language $L$ (see, e.g., \cite{jukna2015}).
We give an example for an exponential gap between these complexities which is due to Jukna \cite{jukna16count}.

Consider the complete bipartite $n\times n$ graph $K_{n,n}$ and let $U,V$ be its two sets of vertices.
A \emph{quasi matching} is a subgraph of $K_{n,n}$ that is constructed by picking an incident edge for each vertex $u \in U$ and each vertex $v\in V$ and then take the union of these two sets.
Let $\QM$ be the set of all quasi matchings, each viewed as characteristic vector of its set of edges.
Then the lower bound $\arithf{\QM} = 2^{\Omega(n)}$ is known \cite[Sect.\:3.1]{jukna16count} (this holds even for circuits).
On the other hand, we have $\boolf{\QM} \leq O(n^2)$:
For an edge $(u,v) \in U\times V$ let $x_{u,v}$ be the corresponding variable. 
Then the boolean function corresponding to $\QM$ can be computed by the following monotone boolean formula of size $O(n^2)$.
\begin{align*}
f^{\textsf{mon}}_{\QM} (x) 
= \left( \bigwedge_{u \in U} \bigvee_{v' \in V} x_{u,v'} \right) \land \left( \bigwedge_{v \in V} \bigvee_{u' \in U} x_{u',v} \right)  \,.
\end{align*}
\end{rem}

\goodbreak
\begin{rem}[Formula	vs.\ circuit complexity] \label{rem:circuits}
Let us emphasize the use of formula size over circuit size in \cref{thm:arith}. For a language $L\subseteq \{0,1\}^n$ let $\arith{L}$ denote the size of a smallest monotone arithmetic circuit producing $L$; for a proper introduction of circuits see \cref{sec:semirings}.
Since every formula is a circuit, the inequality stated in \cref{thm:arith} still holds if we replace formula size $\arithf{L }$ by circuit size $\arith{L}$.
However, with circuit size it is not possible to show superpolynomial blow-ups between finite automata and expressions since automata can be simulated by circuits.

Let $L \subseteq \{0,1\}^n$ be a language that can be accepted by a layered NFA (as defined in \cref{sec:conversions}).
From this NFA we obtain a \emph{monotone algebraic branching program} (MABP) of same size  via the following transformation:
Replace all edges between two neighboring layers $Q_{j-1}$ and $Q_{j}$ labeled with the letter~$0$ (resp. the letter~$1$) by the constant~$1$ (resp. by the variable~$x_j$).
The resulting MABP then describes a polynomial $f(x)= \sum_{a \in L} \const{a} \prod_{i=1}^n x_i^{a_i}$ with $L$ as set of exponent vectors.
By a standard construction, this MABP can be transformed into a monotone arithmetic circuit of polynomial size (see, e.g., \cite{nisan91}).
Thus, if a language $L$ can be described by small NFAs, then the monotone arithmetic circuit complexity of $L$ is also small.

\end{rem}

\section{Applications and limits of the arithmetic bound} \label{sec:app-arith}

\subsection{Uniform languages} \label{sec:uniform}

Recall that a language $L \subseteq \{0,1\}^n$ is \emph{uniform} if all words in $L$ have the same number of ones.
The most basic uniform language is the \emph{binomial language} \[ \binomial{n}{k}= \big\{w \in \{0,1\}^n : |w|_1=k \big\} \]
which was investigated by Ellul et al. \cite{ellul2004}. 
They constructed an expression of length $ n^{O (\log k) }$ (resp. length $O(n \log^k (n))$ if $k$ is a constant) and asked whether its length is optimal. 
Recently, Mousavi \cite{mousavi2017} showed the optimality for $k\leq 3$ by analyzing a linear program derived from the language.
We show that the length is asymptotically optimal also for $k=n^{\Theta(1)}$, using a lower bound shown by \Hrubes\ and Yehudayoff \cite{yehudayoff2011} for the corresponding arithmetic formula complexity. 
Note that the binomial language can be accepted by a DFA of width $k\!+\!1$ and length $n$ (see \cref{fig:binom}), and thus, also \cref{prop:dfa2regex} implies an upper bound of $\rpn{ \binomial{n}{k}} \leq O( k n^{1+\log (k+1)})$.
\begin{cor}[Binomial language] \label{cor:binomial} 
Let $n,k \in \NN$ with $k\leq n/2$. Then the binomial language $\binomial{n}{k}$ requires regular expressions of length $\rpn{ \binomial{n}{k} } \geq n k^{\Omega(\log k)}$.
\end{cor}
\begin{proof}
 The \emph{elementary symmetric polynomial} $f_{n,k}=f_{n,k}(x_1, \dots, x_n)=\sum_{i_1 < i_2 < \dots< i_k} x_{i_1} x_{i_2} \cdots x_{i_k} $ has exactly the words of the binomial language $\binomial{n}{k}$ as its exponent vectors. In \cite[Thm.\ 1\,(a)]{yehudayoff2011} \Hrubes\ and Yehudayoff showed that any monotone arithmetic formula computing a polynomial with the same set of exponent vectors as $f_{n,k}$ has size at least $n k^{\Omega(\log k)}$.
By \cref{thm:arith} the same bound holds for the expression length of $\binomial{n}{k}$.
\end{proof}

\begin{figure}[t]
\begin{center}
\begin{tikzpicture}[xscale=1,semithick]
\node[state,initial] (0) at(0,0) {0};
\node[state]  (1) at (2,0) {1};
\node[state]  (2) at (4,0) {2};
\node (3) at (5.5,0) {\large \dots};
\node[state,accepting] (5) at (7,0) {$k$};

\draw (0) edge[->, bend left=00]  node[above,scale=0.9] {1} (1);
\draw (1) edge[->, bend left=00]  node[above,scale=0.9] {1} (2);
\draw (3) edge[->, bend left=00]  node[above,scale=0.9] {1} (5);
\draw (2) edge[->, bend left=00]  node[above,scale=0.9] {1} (3);

\draw (1) edge[->,loop above,yscale=0.9,xscale=1.20]  node[above,scale=0.9] {0} (1);
\draw (2) edge[->,loop above,yscale=0.9,xscale=1.20]  node[above,scale=0.9] {0} (2);
\draw (5) edge[->,loop above,yscale=0.9,xscale=1.20]  node[above,scale=0.9] {0} (5);
\end{tikzpicture}
\vspace*{-2mm}
\end{center}
\caption{A (partial) DFA accepting the language of all binary words with exactly $k$ ones. 
Thus, the associated $n$-slice DFA accepts the binomial language $\binomial{n}{k}$ and has width $k+1$ and length $n$.}\label{fig:binom}
\end{figure}
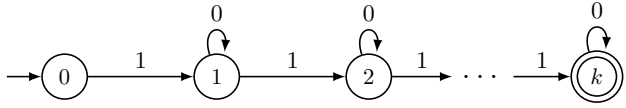

\begin{rem} \label{rem:threshold}
The corresponding monotone boolean function to $\binomial{n}{k}$ is the \emph{threshold function} $\tau_{n,k} : \{0,1\}^n \to \{0,1\}$ that outputs $1$ iff there are at least $k$ ones among $x_1, \dots ,x_n$.
Valiant \cite{valiant84} has shown that this function can be computed by monotone boolean formulas of size $O(n^{5.3})$. So, by the boolean methods (see \cref{sec:related}) only lower bounds of polynomial size can be obtained for $\binomial{n}{k}$.
\end{rem}

The following corollary tells us, that \emph{any} uniform language requires long expressions, provided that it contains sufficiently many words.
\begin{cor}[Uniform languages] \label{cor:uniform}
Let $k\leq n/2$ and $L\subseteq \binomial{n}{k}$ be a uniform language. 
Then $L$ requires regular expressions of length $ \rpn{ L } \geq n k^{\Omega(\log k) } \cdot |L|/\textstyle{\binom{n}{k}}$.
\end{cor}
\begin{proof} \Hrubes\ and Yehudayoff \cite[Prop.\ 7]{yehudayoff2011} showed that any polynomial $g$ with a set $A_g \subseteq \binomial{n}{k}$ of exponent vectors requires monotone arithmetic formulas of size at least $n\:\!k^{\Omega(\log k)}  \cdot |A_g| / \binom{n}{k}$, where $|A_g|$ is the number of monomials in $g$. 
Since each monomial of $g$ corresponds to an exponent vector in $A_g$ and, hence, to a word in $L$, the claim follows by \cref{thm:arith}. 
\end{proof}

A \emph{Dyck word} is a word $w \in \{0,1\}^*$ such that $w$ has the same number of zeros and ones, and every prefix of $w$ contains not more ones than zeros. 
The language $D$ of all Dyck words can  be defined recursively as follows: 
\begin{itemize}[noitemsep,topsep=3pt]
\item $\epsilon \in D$, and
\item if $v$ and $w$ are in $D$, then so are $vw$ and $0 \,w\, 1$.
\end{itemize}
Usually, $D$ is interpreted  as the language of all correctly nested sequences of brackets, with $0$ representing opening and $1$ representing closing brackets. 
It is well known that $D$ is context-free, but not regular. However, if we restrict the length or the height of the Dyck words, the language turns regular.
(The \emph{height} of a Dyck word $w$ is the minimal number $h$ such that $|w'|_0-|w'|_1\leq h$ holds for every prefix $w'$ of $w$.)

For an integer~$n$ consider the $2n$-slice $\dyck{2n}  \isdef  D \cap \{0,1\}^{2n}$ consisting of  all Dyck words of length $2n$. The language $\dyck{2n}$ can be accepted by a DFA with $O(n^2)$ states, see \cref{fig:dyck}. 
In contrast it has no short expressions.
\begin{cor}[Dyck language] \label{cor:dyck} 
The language $\dyck{2n}$ requires regular expressions of length $\rpn{\dyck{2n}} \geq  n ^{\Omega( \log n) } $.
\end{cor}
\begin{proof}
Clearly~$\dyck{2n} \subseteq \binomial{2n}{n}$ holds, since every word in~$\dyck{2n}$ contains the same number of zeros and ones. It is well known that $\dyck{2n}$ contains exactly $\frac{1}{n+1} \binom{2n}{n}$ words -- this number is known as the Catalan number -- see, for example, \cite{chungfeller} for an elegant proof.
Thus, by \cref{cor:uniform}, we have $\rpn{\dyck{2n}} \geq \frac{2n}{n+1}  \cdot n^{\Omega(\log n ) }$.
\end{proof}

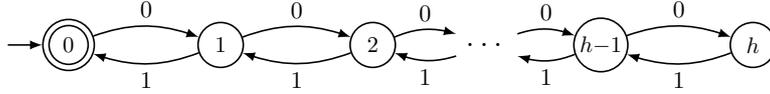
\begin{figure}[t]
\begin{center}
\begin{tikzpicture}[xscale=1,semithick]
\node[state,initial,accepting] (0) at(0,0) {0};
\node[state]  (1) at (2,0) {1};
\node[state]  (2) at (4,0) {2};
\node (3) at (5.5,0) {\large \dots};
\node[state] (4) at (7,0) {$h{-}1$};
\node[state] (5) at (9,0) {$h$};

\draw (0) edge[->, bend left=20]  node[above,scale=0.9] {0} (1);
\draw (1) edge[->, bend left=20]  node[above,scale=0.9] {0} (2);
\draw (3) edge[->, bend left=20]  node[above,scale=0.9] {0} (4);
\draw (2) edge[->, bend left=20]  node[above,scale=0.9] {0} (3);
\draw (4) edge[->, bend left=20]  node[above,scale=0.9] {0} (5);

\draw (1) edge[->, bend left=20]  node[below,scale=0.9] {1} (0);
\draw (2) edge[->, bend left=20]  node[below,scale=0.9] {1} (1);
\draw (3) edge[->, bend left=20]  node[below,scale=0.9] {1} (2);
\draw (5) edge[->, bend left=20]  node[below,scale=0.9] {1} (4);
\draw (4) edge[->, bend left=20]  node[below,scale=0.9] {1} (3);
\end{tikzpicture}
\vspace*{-2mm}
\end{center}
\caption{A (partial) DFA accepting the language of all Dyck words of height at most $h$.  Since in the $2n$-slice $\dyck{2n}$ the height of every word is at most $n$, if we set $h=n$, then the associated $2n$-slice DFA accepts the language $\dyck{2n}$ and has width $n+1$ and length $2n$.
So, the upper bound $\rpn{\dyck{2n}} \leq O( n^{3+\log n})$ follows from \cref{prop:dfa2regex}.}\label{fig:dyck}
\end{figure}

\subsection{Blow-up of language operations} \label{sec:ops}

A classical question is to determine by how much expression length can increase when performing operations like complementation, intersection or shuffle.
The situation for \emph{infinite} languages has been resolved by Gelade and Neven \cite{gelade12} resp.\ Gruber and Holzer \cite{gruber2008ops1,gruber2009ops2}: the blow-up is exponential for shuffle and intersection and double-exponential for complementation. 
For finite languages, however, the blow-up of intersection and shuffle is at most $n^{O( \log n)}$, this can be shown as follows. 
Given two expressions of lengths $m_1$ resp.\ $m_2$ describing finite languages, transform them into NFAs. Then construct the corresponding (intersection or shuffle) product automaton. 
For either of the two operations, this gives an NFA with $n= O(m_1  m_2)$ states. 
Finally, translate this NFA back into an expression. 
Since its accepted language is finite, length $n^{O(\log n ) }$ suffices according to \cref{prop:dfa2regex-old}. 

We now give matching lower bounds, thereby showing that this construction cannot be substantially improved.
Recall the definition of the \emph{shuffle} operation~$\shuffle$ (also called  \emph{interleaving}): For two words~$v,w$, their shuffle~$v \shuffle w$ is the set of all words of the form~$v_1 w_1 v_2 w_2 \cdots v_k w_k$ where~$k\in \NN$,~$v_i,w_i \in \Sigma^*$ for all~$i$ and~$v_1 v_2 \cdots v_k=v$ and~$w_1 w_2 \cdots w_k= w$.  The shuffle of two languages is $L_1 \shuffle L_2 =\bigcup_{ v \in L_1, w\in L_2 } v \shuffle w$.
\goodbreak
\begin{thm}[Blow-up of intersection and shuffle] \label{thm:ops}
There are finite languages $L_1, L_2$ with regular expressions of length $O(n)$ such that
\begin{enumerate}
\item[\textup{(a)}] $\rpn{L_1 \cap L_2} \geq n^{\Omega(\log n)}$,
\item[\textup{(b)}] $\rpn{L_1 \shuffle L_2 }\geq n^{\Omega(\log n)}$.
\end{enumerate}
There is a regular language $L \subseteq \Sigma^*$ with regular expressions of length $O(n)$, such that
\begin{enumerate}
\item[\textup{(c)}] $\rpn{L \cap \Sigma^n} \geq n^{\Omega(\log n)}$.
\end{enumerate}
\end{thm} 
In particular, claims (a) and (b) answer a question asked in \cite{gruber2008cc}, while claim (c) answers a question in \cite[Open Problem 5]{ellul2004} asking for the blow-up of $n$-slices.
Note the difference between claims (a) and (c): in claim (c) we require $\rpn{L}$ to be linear in the length of the words in the intersection $L \cap \Sigma^{n}$, but do not require~$L$ to be finite.

\begin{proof} 
(a) Consider the language~$L_1=(0+\epsilon)^m (1 (0+\epsilon)^m )^{m}$ of all words with exactly $m$ ones and not more than $m$ zeros in a row, and let~$L_2 =(0+1)^{2m}$. 
If we set~$n=m^2$, then both~$L_1$ and~$L_2$ can be described by expressions of length~$O(n)$. The intersection~$L_1 \cap L_2$ is exactly the binomial language~$\binomial{2m}{m}$ and by \cref{cor:binomial} the lower bound $\rpn{L_1 \cap L_2} =\rpn{\binomial{2m}{m}} \geq m^{\Omega( \log m)} = n^{\Omega( \log n) }$ follows.

(b) Consider the languages~$L_1= 0^{n} $ and~$L_2=1^{n}$ with expressions of length~$O(n)$. 
Their shuffle~$L_1 \shuffle L_2$ is exactly the binomial language~$\binomial{2n}{n}$, and \cref{cor:binomial} yields~$\rpn{L_1 \shuffle L_2} \geq n^{\Omega (\log n)}$.

(c) Let~$\Sigma=\{0,1\}$ and~$L=0^* (1 0^*)^{\lfloor n/2\rfloor}$ be the language of all words with exactly $\lfloor n/2\rfloor$ ones; hence,~$\rpn{L} = O(n)$. 
Then~$L \cap \Sigma^n = \binomial{n}{\lfloor n/2 \rfloor}$ and \cref{cor:binomial} yields~$\rpn{L \cap \Sigma^n } \geq n^{\Omega (\log n)}$.
\end{proof}

\subsection{Limitations of the arithmetic bound} \label{sec:limits}
Let us address weaknesses of \cref{thm:arith}.
We already mentioned in \cref{rem:perm} that arithmetic complexity ignores the \emph{order} of the variables.
This prevents us from proving bounds for languages~$L$ that have a permutation $\sigma$ such that $\sigma(L)$ has short expressions.
Take, for example, the language $L=\Lpalin{2n}$ of all palindromes over $\{0,1\}$ of length $2n$. The fooling set method gives an exponential lower bound $\rpn{L} \geq \Omega(2^{n})$ (see \cref{ex:fooling}).
However, reordering the letters yields the language $L'  \isdef \{w_1w_{2n} w_2 w_{2n-1} \dots w_{n} w_{n+1} : w \in L\} = \{00,11\}^{n}$ with $\rpn{L'}=O(n)$. 
According to \cref{rem:perm}, lower bounds obtained by \cref{thm:arith} are the same for $\rpn{L}$ and $\rpn{L'}$, and, thus, are at most linear. In other words, $\arithf{L} \leq O(n)$, but $\rpn{L} \geq \Omega(2^n)$.
The same problem arises for the boolean methods \cite{ellul2004,gruber2008cc} (see \cref{sec:related}), Gruber and Johannsen \cite{gruber2008cc} actually presented the same example.
 
More subtle problems of ``ignored non-commutativity'' can occur even for languages that are invariant under permutations (like the binomial language):
If we transform an expression $R$ into its arithmetic version $\Phi_R$ (as described before \cref{lem:trans}), this formula $\Phi_R$ has a special structure: For every multiplication gate $u \cdot v$, the variables of the gates $u$ and $v$ must be ``consecutive'' in that $u$ contains only variables from the set $\{x_k, \dots, x_j\}$ and $v$ contains only variables from the set $\{ x_{j+1}, \dots ,x_\ell \}$ for some $k \leq j \leq \ell$.
In contrast, an \emph{arbitrary} monotone arithmetic formula (not derived from an expression) can have multiplications of arbitrary sets of variables, e.g., $(x_1 + x_3) \cdot (x_2 + x_4)$ is possible.

Finally, the arithmetic bound only works for languages over the alphabet $\{0,1\}$.
In the next section, we circumvent these issues by translating a lower bound method from arithmetic formula complexity \emph{directly} to regular expression length.

\section{Direct lower bounds} \label{sec:direct}

Until now we only used \emph{already existing} bounds from arithmetic formula complexity. But how can such bounds be obtained? 
One possibility is to lower bound the \emph{circuit depth}. 
By standard balancing arguments a lower bound $2^{\Omega(\mathsf{depth})}$ on formula size follows. Shamir and Snir \cite{shamir} and Tiwari and Tompa 
 \cite{tiwari} developed techniques for such depth bounds, Jukna \cite{jukna2015} put these into a general framework. 
Another option is the lower bound method for the size of monotone (or even multilinear) arithmetic formulas by Shpilka and Yehudayoff \cite{amir10:survey} resp. \Hrubes\ and Yehudayoff \cite{yehudayoff2011} by so-called log-product polynomials.
With this method \Hrubes\ and Yehudayoff showed their lower bound 
for the elementary symmetric polynomial that we used in \cref{cor:binomial} for the binomial language.

So, in order to obtain a lower bound on the expression length of a given language $L$, one can lower bound the monotone arithmetic formula size of $L$ with one of the above methods and apply the arithmetic bound (\cref{thm:arith}).
However, we here take a different approach: we translate the ``log-product method'' from \cite{amir10:survey,yehudayoff2011} \emph{directly} to  expression length.
Thereby we solve the problem of ``ignored non-commutativity'' discussed in \cref{sec:limits} and can use arbitrary alphabets.
A similar translation of lower bound methods to the non-commutative world was done in \cite{hrubes10,non-commuting} (see also \cite{filmus2011} and \cite{seiwert20}), the difference is that we here do it for \emph{formulas} instead of circuits.

\subsection{The log-product bound}\label{sec:logprod}
In the following, $\Sigma$ is an arbitrary alphabet, $L \subseteq \Sigma^n$ a homogeneous language and $R$ a homogeneous expression describing $L$.
The high-level idea for lower bounding $\rpn{L}$ is roughly as follows.
\begin{itemize}[noitemsep]
\item Write $R$ as \emph{union} $B_1+ \dots +B_\ell$ of ``log-product'' expressions~$B_i$, where $\ell\leq \rpn{R}$ and every log-product expression $B_i$ can be factorized into $m \geq \log n$ nontrivial factors $B_i \equiv F_1 F_2 \cdots F_m$.
\item From $L$ derive \emph{structural properties} that any language described by a factor $F_j$  must have.
\item Upper bound the \emph{number of words} in any language with these properties to obtain an upper bound on $|L(B_i)|$.
\end{itemize}

Recall that the degree $\deg R$ of a homogeneous expression $R$ is the length of its described words.
\begin{defi}[Log-product] \label{def:logprod}
A homogeneous expression $B$ is \emph{log-product}\footnote{In the preliminary version of this paper \cite{CS20} we used the term ``balanced'' instead of ``log-product'', and used a slightly different definition.}, if 
\begin{itemize}[topsep=1pt,noitemsep]
\item $B$ is a letter, or
\item if there are homogeneous expressions $B_1, B_2$ such that $B_1$ is log-product itself, $\deg B_1 \geq \deg B_2$ and $B = B_1 B_2$ or $B = B_2 B_1$.
\end{itemize}
\end{defi}
In other words, any single letter is log-product, and if an expression $B_1$ is log-product, then so are $ B_1 B_2$ and $B_2 B_1$ for any homogeneous expression $B_2$ with $\deg B_2\leq \deg B_1$.
For example, the expression $(00+11)(00+11)1$ is log-product, but the expression $(00+11)(00+11)$ is not. In \cref{sec:factor} we will investigate some useful properties of log-product expressions.

The following lemma is a straightforward adaption of \cite[Lem. 4]{yehudayoff2011} or \cite[Lem. 3.5]{amir10:survey}.
\begin{lem} \label{lem:decomp}
Let $R \neq \epsilon$ be a homogeneous expression. 
Then there exist $\ell \leq \rpn{R}$ log-product expressions $B_1, \dots, B_{\ell}$ such that $R \equiv B_1 + \dots + B_{\ell}$.
\end{lem}
\begin{proof}
We proceed by induction on $R$.
If $R$ is a single letter, the claim is trivial.
If $R=R_1 + R_2$ is a union, we can apply the induction hypothesis to both $R_1$ and $R_2$, and are finished.
Finally, let $R=R_1 \cdot R_2$ be a concatenation. 
Assume $\deg R_1 \geq \deg R_2$, the other case is analogous.
By induction hypothesis there are log-product expressions $B_1, \dots, B_{\ell}$ such that $R_1 \equiv B_1+ \dots + B_{\ell}$ for an $\ell \leq \rpn{R_1}$. Since $\deg B_i = \deg R_1 \geq \deg R_2$ holds for all $i$,  every expression $ B_i R_2$ is also log-product. So,
$R \equiv  B_1 R_2+ \dots + B_{\ell} R_2$ is  a union of $\ell \leq \rpn{R_1} \leq \rpn{R}$ log-product expressions, as desired.
\end{proof}

\begin{rem} \label{rem:square}
\cref{lem:decomp} still holds if we extend regular expressions by a \emph{squaring operation} $( \!\,^2)$ defined by $L(R^2)  \isdef  (L(R))^2$ as introduced in~\cite{meyer72}; see also \cite{holzer11} for a more recent overview. To show this, proceed analogously to the case when $R$ is a concatenation of two identical subexpressions $R=R' \cdot R'$.
\end{rem}
From \cref{lem:decomp} our second lower bound method follows.
\begin{thm}[Log-product bound] \label{thm:logprod} 
Let $\Sigma$ be an alphabet, $L \subseteq \Sigma^n$ be a homogeneous language and $h\in \RRpos$. 
If $|L(B)| \leq h$ holds for every log-product expression $B$ with $L(B) \subseteq  L$, 
then any expression for $L$ has length $$\rpn{L} \geq |L| /h\,.$$
\end{thm}
\begin{proof}
Let $R$ be an expression of length $\ell$ for $L$ and assume that $|L(B)| \leq h$ holds for all log-product expressions $B$ with $L(B) \subseteq L$. By \cref{lem:decomp} the expression $R$ can be written as union of at most $\ell$ log-product expressions. At least $\ell \geq |L|/h$ such expressions are necessary to describe all words in $L$. 
\end{proof}

\begin{rem}[Languages with weightings] \label{rem:weightings}
It is possible to generalize \cref{thm:logprod} in the following way.
Let $\lambda : L \to \RRpos$ be a weighting of the words in $L$ and let $\lambda(L')  \isdef  \sum_{w \in L'} \lambda(w)$ for every language $L'\subseteq L$.
In \cref{thm:logprod}, we can replace $| \cdot |$ by $\lambda( \cdot)$:
If $\lambda(L(B)) \leq h$ holds for every log-product expression $B$ with $L(B)\subseteq L$, then
$\rpn{L} \geq \lambda(L) / h$ follows.
However, in this paper we will not use this generalization.
\end{rem}

\begin{rem}[Sublanguages]\label{rem:sublang}
If a lower bound $\rpn{L} \geq \ell$ for a language $L$ is shown by \cref{thm:logprod}, then for every sublanguage $L' \subseteq L$ a lower bound of $\rpn{L'} \geq \ell \cdot |L'|/|L|$ follows.
\end{rem}

\goodbreak
\subsection{Factorizations of log-product expressions} \label{sec:factor}

We take a closer look on log-product expressions.
Given a log-product expression~$B$ of degree $\deg B= n$, we can construct a path from the root to a leaf by always continuing with the child whose subexpression is log-product and that has larger degree (i.e., $B_1$ in \cref{def:logprod}).\footnote{It may happen that this node is not unique. In this case, choose any of the nodes.}  
We call this path the \emph{canonical path} of $B$; note that all inner nodes on this path are concatenation nodes and that all their subexpressions are log-product.
From the canonical path we obtain a factorization 
$$B = F_1 F_2 \cdots F_m $$
of $B$  with $m \geq \log n$ \emph{factors} $F_i$ being some homogeneous expressions, namely the subexpressions of the siblings of the nodes in the canonical path in $B$, plus the last node in this path; see \cref{fig:factor} for an example.
Note that for any factor $F_i$ of $B$ there are words $x$ and $y$ such that $L(x F_i y) \subseteq L(B)$.

\begin{figure}[t]
\begin{center}
\begin{tikzpicture}[thick, scale=0.65, xscale=0.9,
triangle/.style= {fill=gray!0,draw,thick,regular polygon, regular polygon sides=3,
yshift=-0.25cm,inner sep=0pt,outer sep=-0mm,align=center,minimum width=15.5mm}]
\node  (1a) at (-0.25,0) {};
\node(3a) at (1.5,-2) {}; 
\node (4a) at (0,-4) {}; 
\node (6a) at (-1.25,-6) {}; 
\draw[line width=2.7mm,orange!70,opacity=0.65] (1a)--(3a)--(4a)--(6a);

\node[fill=white,node,scale=0.9,inner sep=1pt] (1) at (-0.25,0) {\Large $\Cdot$};
 
\node[scale=0.9,triangle,minimum width=23mm] (2) at (-2,-2.1) {~ };  
\node[scale=0.9]  at (-2,-2.82) {\small $(a\!\cdot\!(b{+}c))$}; 
\node[fill=white,node,scale=0.9,inner sep=1pt] (3) at (1.5,-2) {\Large $\Cdot$}; 

\node[fill=white,node,scale=0.9,inner sep=1pt] (4) at (0,-4) {\Large $\Cdot$}; 
\node[triangle,scale=0.9] (5) at (3,-4) {~ }; 
\node[scale=0.9]  at (3,-4.6) {\small $(a {+} b)$}; 

\node[fill=white,node,scale=0.9] (6) at (-1.25,-6) {$b$}; 
\node[triangle,scale=0.9] (7) at (1.25,-6) {~ }; 
\node[scale=0.9]  at (1.25,-6.6) {\small $(a{+}c)$}; 

\draw (1) --(2.north);
\draw (1) --(3);
\draw (3) --(4);
\draw (3) --(5.north);
\draw (4) --(6);
\draw (4) --(7.north);

\end{tikzpicture}
\vspace*{-2mm}
\end{center}
\caption{A log-product expression $B= (a \cdot (b+c) ) \cdot ( (b \cdot (a+c) ) \cdot (a+ b) )$ and its canonical path (shaded in orange).
$B$ has the factorization $B= F_1 F_2 F_3 F_4$ where the factors are $F_1=(a\cdot(b + c))$, $F_2= b$, $F_3 = (a+c)$ and $F_4=(a+b)$.
}\label{fig:factor}
\end{figure}
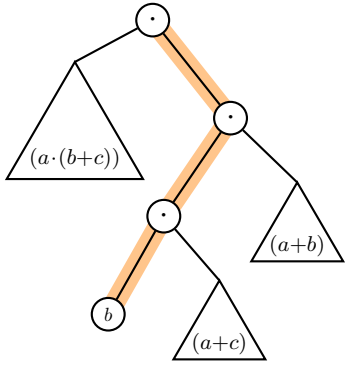

In many cases, factorizations with additional restrictions on the degree of the factors can be useful.
\Hrubes\ and Yehudayoff \cite{yehudayoff2011} used ``balanced'' factorizations of log-product polynomials where each factor $F_i$ has a degree lying between $(1/3)^i \hspace*{1pt} n$ and $(2/3)^i \hspace*{1pt} n$. 
We use factorizations where most factors have ``sufficiently large'' degree.
\begin{prop}[$\gamma$-factorization]\label{prop:factor}
Let $B$ be a log-product expression of degree $n$ and let $\gamma \geq 1$.
Then there exist $m \geq \log (1+n /\gamma)$ and $2m$ homogeneous expressions $P_1, \dots, P_{m} , S_1, \dots, S_{m}$ such that $B \equiv  P_1 P_2 \cdots P_{m}  S_{m} S_{m-1} \cdots S_1$ and
$ \deg P_i + \deg S_i \geq \gamma$ holds for all $i \in [m \!-\!1]$, and $\deg P_{m} +\deg S_{m} \leq  \gamma$.
\end{prop}
We call $P_1 \cdots P_{m}  S_{m} \cdots S_1$ a \emph{$\gamma$-factorization}. Note that trivial factors $P_i=\epsilon$ or $S_i=\epsilon$ are allowed.
\begin{proof}
We proceed by induction.  For $n \leq \gamma$ the claim is trivial, so assume~\mbox{$n >\gamma$}.
To obtain a $\gamma$-factorization of $B$ initialize $S=\epsilon$ and $P=\epsilon$ and follow the canonical path downwards, starting at the root of $B$. 
For each node $v=x\cdot y$ passed by, update either $P$ or $S$: 
if we went to the right child $y$, set $P \isdef  P\cdot B_x$, if we went to the left child $x$, set $S  \isdef  B_y \cdot S$. By this procedure the invariant $B \equiv P B_v S$ holds in each step, where $v$ is the currently reached node. 
Eventually, we arrive at some node $u$ whose degree $n' \! := \deg B_u$ satisfies $(n - \gamma)/2 < n' \leq n - \gamma$.
This is the case just because the degree cannot drop by more than a factor of two at each step. 
Since $B_u$ is log-product, by induction hypothesis there is a $\gamma$-factorization $B_u \equiv P_1 \cdots P_{m'} S_{m'} \cdots S_1$ with $m' \geq \log(1+ n'/\gamma)$.
We claim that $P P_1 \cdots P_{m'} S_{m'} \cdots S_1 S$ is a $\gamma$-factorization of $B$ with $ 2(m'\!+\;\!\!1)$ factors:
since $n'\leq n - \gamma$, we have $\deg P + \deg S = n - n' \geq \gamma$, and
since $n' > (n-\gamma)/2$, we have $m  \isdef  m'+1 \geq \log (1+n' /\gamma)+1  >  \log (1+(n-\gamma)/2\gamma ) +1 = \log (1+n/\gamma )$.
\end{proof}

A log-product expression can also be written as ``balanced'' concatenation of \emph{two} expressions. 
\begin{prop}
\label{prop:factor2}
Let $B$ be a log-product expression of degree $n\geq 2$.
Then there are two homogeneous expressions $X$ and $Y$ such that $B\equiv X\cdot Y$ and
$n/3 \leq \deg X, \deg Y \leq 2n/3 $.
\end{prop}
\begin{proof}
Analogously to the proof of \cref{prop:factor} initialize $P=\epsilon$ and $S=\epsilon$ and follow the canonical path of $B$ downwards, updating $P$ and $S$ in every step such that the invariant $B \equiv P B_v S$ holds.
Eventually, the degree of $P$ or $S$ will jump from $< n/3$ to $\geq n/3$.
At each update step, the degree of $P$ (resp.\ $S$) can increase by at most $(n-\deg P-\deg S)/2$ since we always follow the child of larger degree.
Let $d$ be the degree of $P$ before updating.
Then, after updating, $\deg P \leq d + (n-d-\deg S)/2 = (d+n-\deg S)/2 < (n/3+n -0)/2 = 2n/3$ must hold (resp. $\deg S < 2n/3$ must hold).
This gives us the factorization $B \equiv X\cdot Y$ with $X=P$ and $Y=B_v S$ (resp. $X=P B_v$ and $Y=S$).
\end{proof}

\subsection{Utilizing non-commutativity}  \label{sec:utilize}

In \cref{sec:limits} we discussed that the arithmetic bound (\cref{thm:arith}) is incapable of giving a nontrivial lower bound for the palindrome language 
$L=\Lpalin{2n}= 
\{ww^{\text{reverse}} : w \in \{0,1\}^{n}\}$.
Now, with \cref{thm:logprod} we can give a (suboptimal but nevertheless exponential) lower bound. 
This bound itself is not of interest, we present it only as proof of concept.

Let $n=3k$ for a positive integer $k$.
Take an arbitrary log-product expression $B$ with $L(B) \subseteq  L$. By \cref{prop:factor2} there are homogeneous expressions $X, Y$ such that
$B\equiv X\cdot Y$ and $2n/3 \leq \deg X, \deg Y \leq 4n/3$. Consider a word $w=w_1 \cdots w_{2n}$ described by $X\cdot Y$.
Since $\deg X \geq 2n/3=2k$, all letters $w_{1}, \dots, w_{2k}$ belong to~$X$, and since $\deg Y \geq 2n/3 = 2k$, all letters $w_{4k+1} ,\dots , w_{6k}$ belong to~$Y$. 
By definition of $L$, we must have $w_{i} = w_{6k+1-i}$ for all $i \in [3k]$. 
Thus, all words in $L(X)$ have $w_{6k} w_{6k-1} \cdots w_{4k}$ as prefix and all words in $L(Y)$ have $w_{2k} w_{2k-1} \cdots w_{1} $ as suffix.
Hence, $|L(B)| = |L(X)|\cdot |L(Y)| \leq 2^{k} \cdot 2^{k}  \defis h$, and
\cref{thm:logprod} yields $\rpn{L} \geq |L|/h \geq 2^{3k} / 2^{2k} = 2^k =2^{n/3}$.

\section{Applications of the log-product bound} \label{sec:app-bal}

In the next three subsections we demonstrate applications of \cref{thm:logprod} on the divisibility language $\Ldiv{n}{k}$, the parity language $\Leven{n}{k}$ and the permutation language $\Lperm{n}$.
We will abbreviate the number $|L(R)|$ of described words of an expression $R$ to $|R|$.

\subsection{The divisibility language} \label{sec:div}
Let $p$ be an odd integer. Ellul et al.\ \cite{ellul2004} considered the language of all binary numbers that are divisible by $p$. 
This language has small DFAs with just $p$ states (see \cref{ex:div}), but it seems that expressions must be large.  
However, no lower bound is known so far.
Here we consider the $n$-slice of this language. 
For a word $w\in \{0,1\}^*$ denote its interpretation as binary number by $\bin{w}$; we assume that the most significant bit is the leftmost letter, for example $\bin{0101}=5$, and for convenience let $\bin{\epsilon}=0$. 
The \emph{divisibility language} \[\Ldiv{n}{p}=\big\{ w\in \{0,1\}^n : \bin{w} \equivp 0 \;\! \big\}\] consists of all binary numbers with $n$ bits that are divisible by $p$.
This language also has small DFAs with $O(np)$ states, and expressions of length $\rpn{\Ldiv{n}{p}} \leq O(np \cdot p^{\log (n / \log p)})$ (see below). 
So, the following lower bound is tight, apart from small polynomial factors.

\begin{thm}[Divisibility language] \label{thm:div}
Let $p>2$ be odd. 
Then any regular expression describing the divisibility language $\Ldiv{n}{p}$ has length at least
\[\Rpn{ \Ldiv{n}{p}} \geq \Omega \:\! \big( n^{-1} p^{\log ( n/ \log p ) -2 } \big) \; .\]
\end{thm}
In particular, if $p$ is constant, then $ \Omega ( n^{-1} p^{\log n} ) \leq \rpn{\Ldiv{n}{p}} \leq O(n \:\! p^{\log n})$ holds. 
\begin{proof}
Let $L\;\!\!=\;\!\!\Ldiv{n}{p}$ and $B$ be any log-product expression with $L(B) \subseteq L$.
We~will show an upper bound $h = 2^n n  \cdot p^{-\log (n/ \log p)+1}$ on the number $|B|$ of words in $L(B)$. 
Every $p$-th natural number (beginning with~$0$) is divisible by~$p$, so there are $|L|\geq 2^n/p$ words in $L$. 
Hence, the claimed bound $\rpn{L} \geq |L|/h \geq 2^n/p \cdot  2^{-n} n^{-1}p^{\log (n/ \log p)-1} = \Omega ( n^{-1} p^{\log (n/ \log p) -2 })$  will follow by \cref{thm:logprod}.

For $\gamma  \isdef  \log p$ let $P_1 \cdots P_m S_m \cdots S_1$ be a $\gamma$-factorization of $B$ ensured by \cref{prop:factor}. To prove the bound on $|B|$ we upper bound the number of words described by each factor. 
For $r \in \{0,1,\dots ,p\!-\!1\}$ and $d \in \NN$ let $\Ldivv{d}{p}{r} \isdef \{ w \in \{0,1\}^d: \bin{w} \equivp r \}$ be the language of all $d$-bit numbers that have remainder~$r$ when divided by $p$, for example $\Ldivv{n}{p}{0}= \Ldiv{n}{p}$.
For a word $w$ with $\bin{w} \equivp r $ and a word $x$ of length $d$ with $\bin{x} \equivp r' $, their concatenation satisfies $\bin{wx} \equivp r \cdot 2^d +r'$. 
Since $p$ is odd, the mapping $r \mapsto r\cdot 2^d \bmod p$ is a bijection over $\{0,1,\dots, p  - 1\}$ for every~$d$. 
Thus,  if $\bin{w} \not \equivp \bin{v} $ holds for two words $w$ and $v$ of same length, then also $\bin{wx} \not \equivp \bin{vx} $ and $\bin{xw} \not \equivp \bin{xv} $ hold for any word $x$.

Call a homogeneous expression $T$ \emph{pure}, if all words in $L(T)$ have the same remainder when divided by $p$, that is, if $L(T) \subseteq \Ldivv{d}{p}{r}$ holds for some~$r$ and $d$. 
\begin{clm} \label{clm:pure-div}
Every factor $F$ of $B$ is pure. 
\end{clm} 
\begin{proof}[Proof of \cref{clm:pure-div}]
Assume to the contrary that some factor $F$ of $B$ is not pure, i.e., there are words $w,v \in L(F)$ with $ \bin{w} \not \equivp \bin{v} $. 
Since $F$ is a factor of $B$, there are words $x,y$ such that $ x w y , x v y \in  L(B) \subseteq L$. 
The observation above yields $\bin{x w y} \not \equivp \bin{x v y} $. 
But all words in $L$ must have the same remainder $r=0$, a contradiction.
\clmproof{clm:pure-div}
\end{proof}
Let  $F_i  \isdef  P_i \cdot S_i$ and  $d_i  \isdef  \deg F_i = \deg P_i + \deg S_i$. 
The argument above also implies that concatenations of pure expressions are pure themselves, so every expression $F_i$ is pure.
For all $d$ and all $r$ we have $ \big|\Ldivv{d}{p}{r} \big|\leq  \big\lceil 2^d / p \big\rceil \leq  2^d / p +1$, and
since every expression $F_i$ describes a subset of some language $\Ldivv{d_i}{p}{r}$, the inequality $|F_i| \leq  2^{d_i}/p +1$ holds.
\cref{prop:factor} ensures that $d_i \geq \gamma = \log p$ holds for all $i \in [m\!-\!1]$  and therefore $|F_i| \leq 2^{d_i} / p+1  \leq 2 \cdot  2^{d_i} /p $. Further, since our alphabet is binary, trivially  $|F_m| \leq 2^{d_m}$ must hold.
Finally, recall that $\sum_{i=1}^m d_i = n$ and $m \geq \log (1+n /\gamma ) \geq \log(n/ \log p)$. 
Hence,
\begin{align*}
~~~~~|B| 
& \,= \, \prod_{i=1}^m  |F_i|
\,\leq\, 2^{d_m} \cdot \! \prod_{i=1}^{m-1}  \medfrac{ 2\cdot 2^{d_i}}{p} 
\,=\,  2^{\sum_{i=1}^{m} d_i} \cdot (2/p)^{m-1} &~\\[2pt]
&\,\leq\, 2^{n} \, (2/p)^{\log (n/ \log p)-1} 
\,\leq\, 2^n n \cdot p^{-\log (n/ \log p)+1}  \;  \defis  h \; .  \tag*{\qedhere}
\end{align*} 
\end{proof} 
\vspace*{1mm}

\paragraph{Upper bound} \label{sec:div-upper}  In \cref{ex:div} we gave an upper bound $\rpn{ \Ldiv{n}{p} }\leq O(np \cdot p^{\log n})$. Now we improve this bound to $O(np \cdot p^{\log (n / \log p)} \:\! )$.
To simplify notation assume that $n$ is a power of $2$. 
For each $r \in \{0,\dots, p - 1\}$ and $d \in \{n, n/2, n/4, \dots\}$ define the expression $\Rdiv{d}{p}{r}$ for the language $\Ldivv{d}{p}{r}$ recursively by
$\Rdiv{d}{p}{r} \isdef  
\sum_{r_1, r_2} \Rdiv{ d/2  }{ p } { r_1 }  \cdot \Rdiv{ d/2 }{ p } { r_2 },$
where the sum ranges over all $p$ combinations $(r_1, r_2)$ with $r_1 2^{  d/2  }+r_2 \equivp r$.
If $d < \log p$, then $\Ldivv{d}{p}{r}$ contains at most one single word $w$, in this case let $\Rdiv{d}{p}{r} \isdef w$.
Finally, $\Rdiv{n}{p}{0}$ describes the language $\Ldivv{n}{p}{0} =  \Ldiv{n}{p}$.
This recursion has $2p$ branches in each step, depth at most $\lceil \log (n / \log p) \rceil$ and every base case expression has length at most $\log p$. 
Thus, length $\rpn{\Ldiv{n}{p}{}}\leq O( \log p \cdot(2p)^{\lceil\log (n / \log p)\rceil}) \leq O(np  \cdot p^{\log (n / \log p)} \:\! )$ suffices.

\paragraph{Blow-up of DFA Conversion} \label{sec:conversion}
Gruber and Johannsen \cite{gruber2008cc} showed that converting a DFA for a finite language into an expression can cause a blow-up of $n^{(\log n)/192}$, this is optimal apart from the factor in the exponent. (They actually stated the factor $1/75$ instead of $1/192$ in \cite[Thm.\ 10]{gruber2008cc}, but there seems to be a minor mistake in their proof regarding the number of states of the DFA.)
We now can improve this factor to $1/4 - o(1)$.
\begin{cor} \label{cor:conversion}
There are finite languages $L_m$ that can be accepted by DFAs with $m$ states, but require regular expressions of length at least $m^{(\log m)/4 - \Theta(\log \log m)}$.
\end{cor}
\begin{proof}
Let $n= p=\sqrt{m} \geq 3$ and $L_m = \Ldiv{n}{p}$.
This language can be accepted by a DFA with $m=n p$ states. 
The slice DFA described for $\Ldiv{n}{p}$ in \cref{ex:div} has at most $p$ states in each layer $Q_j$ for $j \in [n-1]$, and additionally one initial, one final and one trap state.
This leaves us with at most $3 + (n-1)p \leq np = m$ states in total, since $p\geq 3$.
\Cref{thm:div} yields the lower bound 
$\rpn{ \Ldiv{n}{p} } \geq n^{-1} \cdot p^{\log(n/\log p)-2 }$.
For $n=p=m^{1/2}$ we get $\rpn{ L_m } \geq  n^{\log(n) - \Theta(\log \log n) }
= m^{ (\log m) /4- \Theta(\log \log m )}$.
\end{proof}

\subsection{The parity language} \label{sec:par}
Let $\Sigma=\{1,\dots,k\}$ be an alphabet for $k\geq 2$ and $n$ be even. Consider the \emph{parity language} 
$$\Leven{n}{k} =\big\{w \in \Sigma^n: |w|_j \equiv_2 0\textup{  for all } j \in \Sigma \,\big\}$$ 
of all length $n$ words that have an even number of occurrences of every letter. This language naturally generalizes the XOR language $\Lxor{n}=\{w \in  \{0,1\}^n : |w|_1 \text{ is even} \}$ mentioned in the introduction. It can be accepted by DFAs of size $O( 2^k\;\! n)$ and has expressions of length  $O(k \hspace{1pt}2^k n^k)$, see below.  We give an almost matching lower bound.

\begin{thm}[Parity language] \label{thm:even} Let $n$ be even and $k \geq 2$.
Then the parity language $\Leven{n}{k}$ requires regular expressions of length $\rpn{\Leven{n}{k}} \geq    \Omega(
n^{k-2} \, (4k \ln k )^{2-k} ) = n^{k-2} \, k^{-\Theta(k)}$.
\end{thm}
\noindent
In particular, if $k$ is constant, then $ \Omega ( n^{k-2}) \leq \rpn{\Leven{n}{k}} \leq O(n^{k})$ holds. 

To prove \cref{thm:even} it will be convenient to have a look at the minimal DFA $\Aeven{k}$ for the \emph{infinite} variant $\Levenn{k}  \isdef  \{w \in \Sigma^* : |w|_{j} \equiv_2 0\textup{ for all } j \in \Sigma \}$ of the parity language. 
This DFA has states $Q=\{0,1\}^k$ where the $j$-th  bit $q_j$ of a state $q\in Q$ indicates the parity of the number of occurrences of letter $j$.
The initial state and the only final state is the all-$0$ vector $\vec{0}$. 
The transitions are defined as $\delta(q,j)=q + \vec{e}_j \bmod 2$ for all $j \in [k]$ where $\vec{e}_j$ is the $j$-th unit vector. That is, when reading the letter~$j$, the automaton flips the $j$-th bit of the current state. 
The underlying graph of $\Aeven{k}$ is the $k$-dimensional hypercube.
An example is given in \cref{fig:cube}.

\begin{figure}[t]
\begin{center}
\begin{tikzpicture}[scale=0.87,semithick]
\tikzstyle{tstate}=[state,scale=0.92]
\node[tstate,initial,accepting] (000) at(0,0) {\small 000};
\node[tstate]  (001) at (3,0) {\small 001};
\node[tstate]  (010) at (0,3) {\small 010};
\node[tstate]  (011) at (3,3) {\small 011};
\node[tstate]  (100) at (0+1.58,0+1.58) {\small 100};
\node[tstate]  (101) at (3+1.58,0+1.58) {\small 101};
\node[tstate]  (110) at (0+1.58,3+1.58) {\small 110};
\node[tstate]  (111) at (3+1.58,3+1.58) {\small 111};

\draw (000) edge[<->]  node[above=-1pt,scale=0.90] {$3$} (001);
\draw (010) edge[<->]  node[above left=-1pt,scale=0.90] {$3$} (011);
\draw (100) edge[<->]  node[above right=-1pt,scale=0.90] {$3$} (101);
\draw (110) edge[<->]  node[above=-1pt,scale=0.90] {$3$} (111);

\draw (000) edge[<->]  node[left=-1pt,scale=0.90] {$2$} (010);
\draw (001) edge[<->]  node[below left=-1pt,scale=0.90] {$2$} (011);
\draw (100) edge[<->]  node[above right=-1pt,scale=0.90] {$2$} (110);
\draw (101) edge[<->]  node[left=-1pt,scale=0.90] {$2$} (111);

\draw (000) edge[<->]  node[above left=-3pt,scale=0.90] {$1$} (100);
\draw (001) edge[<->]  node[above left=-3pt,scale=0.90] {$1$} (101);
\draw (010) edge[<->]  node[above left=-3pt,scale=0.90] {$1$} (110);
\draw (011) edge[<->]  node[above left=-3pt,scale=0.90] {$1$} (111);
\end{tikzpicture}
\vspace*{-2mm}
\end{center}
\caption{The DFA $\Aeven{k}$ for the language $\Levenn{k}$ for $k=3$.} \label{fig:cube}
\end{figure}
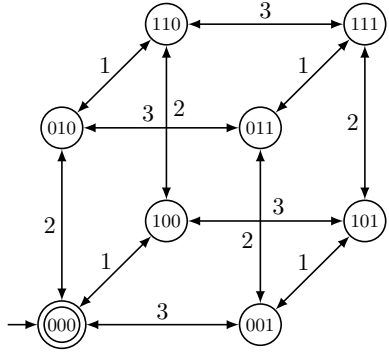

To obtain the claimed upper bound on $\rpn{\Leven{n}{k}}$, consider the $n$-slice DFA of $\Aeven{k}$.
Since for every word $w$ we have $|w| = \sum_{j \in \Sigma} |w|_j$, half of the states in each layer of the slice DFA are unreachable and can be removed.
So we obtain a DFA of width $2^{k-1}$ and length $n$ accepting  $\Leven{n}{k}$. By \cref{prop:dfa2regex} the upper bound $\rpn{\Leven{n}{k}} \leq  O(kn 2^{(k-1)(\log n+1)}) = O(k \hspace{1pt} 2^k n^k)$ follows.

We now turn to the proof of \cref{thm:even}. 
For every $q \in Q$ and $d \in \NN$ let $\Lpar{d}{k}{q}$ be the language of all words $w \in\Sigma^d$ that end in state $q$ when given as input to the DFA $\Aeven{k}$. 
These are exactly those words $w$ for that $|w|_{j} \equiv_2 q_j$ holds for all $j \in \Sigma$. For example $\Lpar{n}{k}{\vec{0}} = \Leven{n}{k}$.
We need an estimate on the number of words in the languages $\Lpar{n}{k}{q}$.

\begin{lem} \label{lem:cube}
Let $n,k\in \NN, k \geq 2$. Then the following hold:
\begin{enumerate}
\item[\textup{(a)}] $ |\Leven{n}{k}| \geq k^n \cdot 2^{1-k}$ for even $n$,
\item[\textup{(b)}] $ |\Lpar{n}{k}{q}| \leq k^n \cdot 2^{2-k}$ for $n \geq k \ln k $ and all $q \in \{0,1\}^k$.
\end{enumerate}
\end{lem}
The proof of this lemma is somewhat technical and is therefore postponed to the end of this subsection. 
\begin{proof}[Proof of \cref{thm:even}]
Let $L=\Leven{n}{k}$ and $B$ be any log-product expression with $L(B) \subseteq L$. 
Our goal is to show an upper bound on the number $|B|$ of words in~$L(B)$. 
For $\gamma  \isdef  k \ln k $ let $P_1 \cdots P_m S_m \cdots S_1$ be a $\gamma$-factorization of $B$ ensured by \cref{prop:factor}. To prove the bound on $|B|$ we upper bound the number of words described by each factor of $B$. 
Call a homogeneous expression $T$ \emph{pure} if its described language is a subset of $\Lpar{d}{k}{q}$ for some $d \in \NN$  and $q \in \{0,1\}^k$.
\begin{clm} \label{clm:pure-even}
Every factor $F$ of $B$ is pure.
\end{clm}
\begin{proof}[Proof of \cref{clm:pure-even}]
Assume to the contrary that some factor $F$ is not pure, i.e., there are words $w,v \in L(F)$ and a letter $a \in \Sigma$ such that $|w|_a$ is even and $|v|_a$ is odd. 
Since $F$ is a factor of $B$, there are words $x$ and $y$ such that $xwy,xvy \in L(B) \subseteq L$. 
But then  $|w'|_a = |x|_a+ |w|_a +|y|_a$ is odd or $|v'|_a = |x|_a+ |v|_a +|y|_a$ is odd, a contradiction to all words in $L$ having an even number of occurrences of every letter.
\clmproof{clm:pure-even}
\end{proof}
Let  $F_i  \isdef  P_i \cdot S_i$ and  $d_i  \isdef  \deg F_i = \deg P_i + \deg S_i$. 
The argument above also implies that concatenations of pure expressions are pure themselves, so every $F_i$ is pure.
For all $i \in [m-1]$, \cref{prop:factor} ensures $d_i \geq \gamma = k \ln k $ and so
\cref{lem:cube}\,(b) yields $|F_i| \leq k^{d_i} 2^{2-k}$. Moreover, $|F_m| \leq k^{d_m}$ trivially holds and
recall that $\sum_{i=1}^m d_i = n$.
Hence,
\begin{align}
|B| &=  \prod_{i=1}^{m} |F_i|
\leq  k^{d_m} \cdot \prod_{i=1}^{m-1} \big(k^{d_i} 2^{2-k} \big)
=   k^{\sum_{i=1}^m d_i} \cdot 2^{(2-k) (m-1)}  
\leq  k^n  \cdot 2^{(2-k)(m-1) } \, \defis h \,.
\end{align}
According to \cref{lem:cube}\,(a) there are at least $|L|\geq k^n \cdot 2^{1-k}$ words in $L$. 
Hence, the bound 
\begin{align}
\rpn{L} \geq |L|/h \geq  2^{1-k +(k-2)(m-1)} =  \Omega(2^{(k-2)(m-2)})
\end{align}
follows by \cref{thm:logprod}.
We have $m \geq \log (1+n /\gamma ) \geq \log(n/ (k \ln k )) =  \log n - \log (k \ln k)$.
Thus, 
\begin{align}
\rpn{L}\geq \Omega( 2^{(k-2)(\log n - \log(k \ln k) -2)} ) = \Omega(  n^{k-2} \cdot (4 k\ln k)^{2-k}) \,.\tag*{\qedhere}
\end{align}
\end{proof}
\vspace*{1mm}

\begin{proof}[Proof of \cref{lem:cube}]
Our goal is to show these two inequalities:
\begin{align}
 \big|L^{\textup{even}}_{n,k}\big| \geq \;&k^n \cdot 2^{1-k} \text{ for even }n, \label{eq:cube-lower}\\[1pt]
 \big|\Lpar{n}{k}{q}\big| \leq \;& k^n \cdot 2^{2-k}  \text{ for }  k \geq 2, n \geq k \ln k \text{ and } q \in \{0,1\}^k .\label{eq:cube-upper}
\end{align}
We analyze the probability of a \emph{random} word $w \in_R \Sigma^n$ being accepted by the DFA $\Aeven{k}$, recall that $\Sigma=[k]$ is our alphabet and that $\Aeven{k}$ has $Q=\{0,1\}^k$ as its set of states.
Let the word $w=w_1 \cdots w_n$ be drawn uniformly at random from $\Sigma^n$; that is, each letter $w_i$ is chosen independently uniform from $\Sigma$. If we travel for each letter $w_i$ along the corresponding edge of $\Aeven{k}$, this process corresponds to a random walk on the $k$-dimensional hypercube $\{0,1\}^k$, starting in the origin~$\vec{0}$. For a state $q$ let $p_q$ be the probability of ending in~$q$ after $n$ steps and let $|q|$ be the number of ones in~$q$.

According to Diaconis, Graham and Morrison\footnote{To be precise, they consider a slightly different random walk with self-loops. This difference only affects the term $2j/k$.} \cite[Lem.\:1]{hypercube2} the probability $p_q$ is given by
\begin{align}
p_q = 2^{-k}  \sum_{j=0}^k \ \left( 1- \frac{2j}{k} \right)^n \,\sum_{i=0}^{|q|} (-1)^i \binom{|q|}{i} \binom{k-|q|}{j-i} \,. \label{eq:pq}
\end{align}

The probability of a word being accepted by $\Aeven{k}$ is exactly the probability $p_{\vec{0}}$ of ending in the origin~$\vec{0}$.
So, to show \cref{eq:cube-lower} we need to lower bound $p_{\vec{0}}$ for even $n$. 
For $q= \vec{0}$ the probability given in \cref{eq:pq} simplifies to $p_{\vec{0}} = 2^{-k}  \sum_{j=0}^k  ( 1- 2j/k)^n \binom{k}{j}$. 
For even $n$, every summand is nonnegative. 
Hence, we can drop all summands with $0<j<n$ to obtain the lower bound
\begin{align}
p_{\vec{0}}
\,&\geq 2^{-k}  \sum_{j \in \{0,k\}} \left( 1- \frac{2j}{k} \right)^n \binom{k}{j} 
= 2^{-k}  \left[ 1 ^n \binom{k}{0} + \left( - 1\right)^n \binom{k}{k} \right]
= 2^{1-k} \;. \label{eq:6}
\end{align}

To show \cref{eq:cube-upper} we need to upper bound the probability $p_q$ given by \cref{eq:pq} for an arbitrary state $q$. 
First, we bound the inner sum using the Chu--Vandermonde identity $\sum_{i=0}^{j} \binom{|q|}{i} \binom{k-|q|}{j-i} / \binom{k}{j} = 1$ (which is clear from interpreting the summands as probabilities of a hypergeometric distribution), and the fact that the binomial coefficient $\binom{k-|q|}{j-i}$ in the inner sum of \cref{eq:pq} is zero for all values $i>j$:
\begin{align}
Z_j \isdef \sum_{i=0}^{|q|} \underbrace{(-1)^i}_{\leq 1} \binom{|q|}{i} \binom{k-|q|}{j-i}
& \;\!\leq\, \sum_{i=0}^{j}  \binom{|q|}{i} \binom{k-|q|}{j-i}
= \binom{k}{j} \;.
\end{align}
Analogously, we obtain $Z_j \geq - \binom{k}{j}$.
Then
\begin{align}
p_q 
&= 2^{-k}  \sum_{j=0}^k  \left( 1- \frac{2j}{k} \right)^n Z_j 
=   2^{-k}  \left[ \sum_{j <k/2}  \!\left( 1- \frac{2j}{k} \right)^n  Z_j+ \sum_{j <k/2}  \!\left(  \frac{2j}{k}-1 \right)^n Z_{k-j}\right] \;.
\intertext{If $n$ is odd, then all summands are positive in the first sum and negative in the second, and we use the bound $Z_j\leq \binom{k}{j}$ in the first sum and the bound $-Z_{k-j} \leq \binom{k}{k-j} = \binom{k}{j}$ in the second sum. If $n$ is even then the summands in both sums are positive and we use the bounds $Z_j \leq \binom{k}{j}$ and $Z_{k-j} \leq \binom{k}{k-j} = \binom{k}{j}$.  In both cases we get}
p_q &\leq   2^{-k}  \sum_{j<k/2} 2\binom{k}{j} \left( 1- \frac{2j}{k} \right)^n 
\leq    2^{1-k}   \sum_{j<k/2} k^j  e^{-2jn/k } 
\leq \,   2^{1-k}   \sum_{j=0}^{\infty} \left( k \:\!e^{-2n/k} \right)^j \\
&=    2^{1-k} \cdot   \frac{1}{1- k\:\!e^{-2n/k}} 
\overset{n \geq k \ln k }{\leq}    2^{1-k}	\cdot \frac{k}{k-1} \overset{k \geq 2}{\leq}   
2^{2-k} \,. \label{eq:10}
\end{align} 
Since there are $|\Sigma^n|=k^n$ possible choices for a word $w \in \Sigma^n$, we have $ |\Lpar{n}{k}{q}| =p_q \cdot k^n$, and the claimed bounds
$ \big|\Leven{n}{k}\big| = k^n p_{\vec{0}} \geq k^n \cdot 2^{1-k}$ and $\big |\Lpar{n}{k}{q}\big| = k^n p_q \leq  k^{n} \cdot 2^{2-k}$ follow from \cref{eq:6,eq:10}.
\end{proof}

\subsection{The permutation language} \label{sec:perm}

Let $n \geq 1$ and consider the language $\Lperm{n}$ of all permutations over the alphabet $\Sigma=[n]$. 
A (partial) DFA for $\Lperm{n}$ with $2^n$ states can be constructed as follows.
Let $Q = \{ q \subseteq [n]\}$ be the set of states, $i=\emptyset$ the initial state, $[n]$ the only final state and the (partial) transition function defined as $\delta(q, a) = q \cup \{a\}$ for all $ a \in \Sigma \setminus q$. An example for $n=3$ is given in \cref{fig:permdfa}.

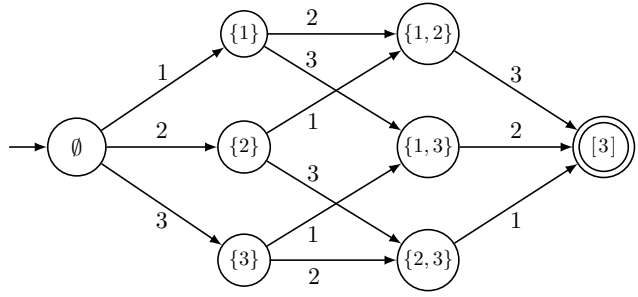
\begin{figure}[t]
\begin{center}
\begin{tikzpicture}[xscale=1.1,semithick,scale=1.00]
\node[state,initial,minimum width=9mm,inner sep=1pt] (0) at(0,1.5) {$\emptyset$};
\node[state,minimum width=8mm,scale=.9]  (1) at (2,3)  {$\{1\}$};
\node[state,minimum width=9mm,scale=.9]  (2) at (2,1.5) {$\{2\}$};
\node[state,minimum width=9mm,scale=.9]  (3) at (2,0)  {$\{3\}$};

\node[state,minimum width=9mm,scale=.9]  (12) at (4.2,3) {$\{1,2\}$};
\node[state,minimum width=9mm,scale=.9]  (13) at (4.2,1.5) {$\{1,3\}$};
\node[state,minimum width=9mm,scale=.9]  (23) at (4.2,0) {$\{2,3\}$};

\node[state,accepting,minimum width=9.5mm,scale=.9] (123) at (6.3,1.5) { $[\hspace*{1pt} 3 \hspace*{1pt}]$};

\draw (0) edge[->, bend left=00]  node[above,scale=0.9] {1} (1);
\draw (0) edge[->, bend left=00]  node[above,scale=0.9] {2} (2);
\draw (0) edge[->, bend left=00]  node[below,scale=0.9] {3} (3);

\draw (1) edge[->, bend left=00]  node[pos=0.33,above,scale=0.9] {3} (13);
\draw (1) edge[->, bend left=00]  node[pos=0.33,above,scale=0.9] {2} (12);

\draw (2) edge[->, bend left=00]  node[pos=0.33,above,scale=0.9] {3} (23);
\draw (2) edge[->, bend left=00]  node[pos=0.33,below,scale=0.9] {1} (12);

\draw (3) edge[->, bend left=00]  node[pos=0.33,below,scale=0.9] {1} (13);
\draw (3) edge[->, bend left=00]  node[pos=0.33,below,scale=0.9] {2} (23);

\draw (12) edge[->, bend left=00]  node[above,scale=0.9] {3} (123);
\draw (13) edge[->, bend left=00]  node[above,scale=0.9] {2} (123);
\draw (23) edge[->, bend left=00]  node[below,scale=0.9] {1} (123);

\end{tikzpicture}
\vspace*{-2mm}
\end{center}
\caption{A (partial) DFA accepting the permutation language $\Lperm{n}$ for $n=3$.}\label{fig:permdfa}
\end{figure}

Jerrum and Snir \cite{jerrum1982} showed a lower bound of $\rpn{\Lperm{n}} \geq 2^n -2$ by a method for (non-commutative) circuit complexity.
Independently, Ellul et al. showed a lower bound of $2^n-1$ \cite[Thm.\:29]{ellul2004} by the fooling set method.
These bounds were recently improved by Molina Lovett and Shallit \cite{shallit2019} to $4^n n^{ - (\log n)/4 + \Theta(1) }$ by a custom argument; they also showed that this bound is tight.
We give an alternative, simpler proof of this latter bound using the log-product bound (\cref{thm:logprod}). 

\begin{thm}[Permutation language, \cite{shallit2019}] \label{thm:perm}
The permutation language $\Lperm{n}$ requires regular expressions of length at least $\rpn{\Lperm{n}} \geq \Omega(4^n n^{- (\log n)/4-3/4})$.
\end{thm}
\noindent
We claim that any log-product expression describing a subset of the permutation language $\Lperm{n}$ describes not more than roughly 
$ (n/2)! (n/4)! (n/8)!\cdots 2! \hspace*{1.5pt} 1!$ words. By the Stirling approximation $n! \sim \sqrt{2\pi n} \hspace*{1pt}(n/e)^n$, this term is approximately $n!\cdot 4^{-n} n^{( \log n)/4 + \Theta( 1)}$.
In the following lemma we make this estimate precise. 
\begin{lem}  \label{lem:perm} Let $n \geq 1$ and $B$ be a log-product expression with $L(B) \subseteq \Lperm{n}$.
Then $B$  describes at most $|B|\leq h(n) \isdef n! \cdot 4^{1-n} n^{\frac14 \left[ 3+\log n \right]}$ words.
\end{lem}
We show first that \cref{thm:perm} follows from this lemma and then give its proof.
\begin{proof}[Proof of \cref{thm:perm}]
Let $B$ be any log-product expression with $L(B) \subseteq \Lperm{n}$. 
By \cref{lem:perm} we have $|B| \leq h(n)$ and there are $n!$ words in $\Lperm{n}$. 
Thus,  \cref{thm:logprod} yields
\usetagform{simple}
\begin{align}
\rpn{\Lperm{n}} 
&\geq  |\Lperm{n}| / h(n) = 4^{n-1} \, n^{-(\log n)/4 - 3/4} \;.  \tag{\qedhere}
\end{align}
\end{proof}

\begin{proof}[Proof of \cref{lem:perm}]
In contrast to the proofs of \cref{thm:div,thm:even} we do not use $\gamma$-factorizations here, but work directly with \cref{def:logprod}.
We proceed by induction. If $n=1$, then only one single word (consisting of a single letter) can be described, and $|B| \leq 1 =1! \cdot 4^0 1^{3/4 + 0/4} = h(1)$, as desired.

Now let $n \geq 2$ and assume that the claim holds for all $k < n$.
By \cref{def:logprod}, there are homogeneous expressions $A, B'$ such that $B=AB'$ or $B=B' A$ holds, $B'$ is log-product and $k \isdef \deg B' \geq \deg A$. We can assume w.l.o.g. that $\deg A = n-k \geq 1$ holds.
Since every word described by $B$ is a permutation, there must be a partition $ \Sigma_A \cupp \Sigma_{B'} = \Sigma$ of the alphabet $\Sigma=[n]$ with $|\Sigma_A| =\deg A, |\Sigma_{B'}| = \deg B' $ such that all words described by $A$ and by $B'$ are permutations over $\Sigma_A$ and $\Sigma_{B'}$, respectively.
Thus, we have $|A| \leq (\deg A)! = (n-k)!$ and, since~$B'$ is log-product, by induction $|B'| \leq h(k)$ holds.
Hence,
\begin{align}
|B|= |A| \cdot |B'|  
\leq (n \minus k)! \, h(k)  
 &= (n \minus k)! \,k ! \cdot 4^{1-k} k^{\frac14\left[3+\log k \right]}  
 = 4n! \, \binom{n}{k}^{\!\!-1} \! \cdot 4^{-k} k^{\frac{1}{4}\left[3+\log k \right]} \;. \label{eq:perm}
 \end{align}
For the cases $n=2, \dots, 7$ we check by hand that $(n - k)! \,h(k) \leq h(n)$ holds, see \cref{table:byhand}.
\begin{table}[t] 
\setlength{\tabcolsep}{2.6mm}
\def\arraystretch{1.05}
\begin{center}
\begin{tabular}{c||l|l|l|l|l|l}
$n$                              &  2 & 3 & 4 & 5 & 6 & 7  \\ \hline 
$\max\limits_{n-1\geq k \geq n/2}^{\phantom{A}} \,(n-k)!\,h(k)$  &  1.00 & 1.00 & 2.00 & $2.64...$ & $7.92...$ & $12.72...$  \\[2pt]
$h(n)$                                  &  1.00 & $1.32...$ & $2.12...$ & $3.89...$ & $8.58...$ & $20.74...$ 
\end{tabular}
\end{center} \vspace*{-3mm}
\caption{For the cases $n=2, \dots, 7$ we see that for all possible values of $k$ indeed $|B| \leq (n-k)!\,h(k) \leq h(n)$ holds. 
}\label{table:byhand}
\end{table}
For the rest of the proof assume $n\geq 8$.
We use the inequality $\binom{n}{k} \geq  2^{n H(x)} \;\!\! / \sqrt{2n}$, 
where $H(x)= - x \log x - (1\minus x) \log (1\minus x)$ is the binary entropy function and $x=k/n$ (see, e.g., \cite[Ch.\:10, Lem.\:7]{sloane}). By applying this inequality to \cref{eq:perm}, we obtain
\begin{align}
|B| &\leq \;\! g(k)  \isdef  4n! \sqrt{2n} \cdot 2^{f(k/n)}  ~~\text{ with }~~
 f(x)  \isdef   -n H(x)  -2nx + \medfrac{1}{4} \big[3 \log(nx)+\log^2(nx) \big]\;. \label{eq:perm-g}
\end{align}
We claim that $f(x)$ is convex on the interval $[1/2,1)$. Intuitively, this is the case because the entropy function $H(x)$ is concave (and, thus, $-H(x)$ is convex) and dominates the other terms. 
To prove this rigorously, we show that the second derivative $f''(x)$ of $f(x)$ is positive for all $x \in [1/2,1)$:
\begin{DispWithArrows}[fleqn,displaystyle]
 \ln(2)\cdot f''(x)
&=  \frac{n}{1-x} + \frac{n}{x} + \frac{-3-2 \log(nx/e)}{4x^2 } \Arrow{~since $1/2\leq x < 1$} \\[2pt]
&>  2n + n    -3- 2\log(n/e)  \Arrow{~since $n \geq 8$} \\[3pt]
&\geq   16 + 8 -3 -6+ 2\log e  \: > \, 0 
\end{DispWithArrows}
Since $f(x)$ is convex for $x \in [1/2,1)$, also the function $g(k)=4n! \sqrt{2n}\cdot 2^{f(k/n)}$ is convex for $k\in[n/2, n-1]$.
Therefore, $g(k)$ achieves its maximum on the boundary, namely for $k=n/2$ or $k=n-1$.
We show that in both cases $g(k) \leq h(n)$ holds, which completes the proof.

\textit{Case 1:} If $k=n-1$, then $x=(n-1)/n$ and we get for entropy
\begin{align}
nH\!\left( \medfrac{n-1}{n} \right)  = -n \left( \medfrac{n-1}{n} \log\medfrac{n-1}{n} + \medfrac{1}{n} \log \medfrac{1}{n} \right)
=  (n-1) \!\! \underbrace{\log \medfrac{n}{n-1}}_{\geq 1/(n-1)~\:\!(\ast)\label{test}} \!\!\! +  \log n 
\geq \log(2n),
\end{align} where inequality $(\ast)$ holds since $\log(1+\epsilon) \geq \epsilon $ holds for all $ \epsilon \in [0,1]$.
Hence,
\begin{align}
g(n-1) 
& = 4n! \sqrt{2n} \cdot \underbrace{\,2^{-n H((n-1)/n)}}_{\leq \:\! 1/2n }\, \underbrace{\,2^{ -2(n-1)}}_{=\:\! 4^{1-n}} \, \underbrace{\,(n \minus 1)^{\frac14 [3 +\log(n-1) ]}}_{\leq \:\! n^{\frac14 [3  +\log n  ]}}  
\,\leq\, \medfrac{4\sqrt{2n}}{2n} \, h(n) \, \overset{n \geq 8}{\leq} \, h(n) \,.
\end{align}

\textit{Case 2:} If $k = n/2$, then $x=1/2$ and $H(1/2)=1$.  Hence,
\begin{align}
g(n/2) 
& = 4n! \sqrt{2n} \cdot \underbrace{\,2^{-n H(1/2)}}_{=2^{-n} }\, 2^{ -n} \,(n/2)^{\frac14 [3 +\log \frac{ n} {2} ]}  
 \,\overset{(\ast \ast)}{=} \,4n! \; 4^{-n} \, n^{\frac14 \left[ 3+ \log n \right] } \, =h(n)  ,
\end{align}
where equality ($\ast \ast$) holds since 
\begin{align*}
\sqrt{2n} \, (n/2)^{\frac14 \left[3+\log \frac{n}{2}\right]} &=
2^{\frac12 \log(2n)}  \,(n/2)^{\frac14 \left[3+\log \frac{n}{2}\right]} \\
&= 2^{\frac14 \left[ 2\log(2n)-3 -\log \frac n 2 \right]} \, n^{\frac14 \left[3+ \log \frac{n}{2} \right] } \\
&= 2^{\frac14 \log n} \, n^{ \frac14 \left[2+ \log n  \right] }  
\,=\, n^{\frac14 \left[3 + \log n \right]}  \;. \tag*{\qedhere}
\end{align*}

\end{proof}

\goodbreak 
\section{Non-homogeneous and infinite languages} \label{sec:envelope}

Until yet, we only dealt with homogeneous languages (in which all words have the same length).
In \cref{thm:arith} this restriction  is crucial, since
for a non-homogeneous expression it is in general not possible to assign a unique position to each leaf.
For example, consider the expression $(0+00)1$ describing the non-homogeneous language $L=\{01, 001\}$. Then the letter~$1$ appears in one word at position $2$ and in the other at position~$3$.
Also \cref{thm:logprod} does not work for non-homogeneous languages since the definition of log-product expressions is based on the degree which is only defined for homogeneous expressions.

Nevertheless, in some cases one can circumvent these problems. 
Following the notation in \cite{jerrum1982}, for a (possibly infinite) language $L$ define its \emph{lower envelope} $\lenv{L}$ as the set of all words in $L$ of minimal length, and for a finite language $L$ define  its \emph{higher envelope} $\henv{L}$ as the set of all words in $L$ of maximal length.
Note that both envelopes are homogeneous languages.

Further, let $\Sigma$ be an alphabet and $\mu : \Sigma \to \RRpos$ be a nonnegative weighting of the letters in $\Sigma$. For a word $w \in \Sigma^n$ its \emph{weight} is $\mu(w) = \sum_{i=1}^n \mu(w_i)$.
Given such a weighting $\mu$, for a language $L$ its \emph{lower $\mu$-envelope} $\lenvv{L}{\mu}$ is the set of all words in $L$ of minimal weight, and
for a finite language $L$ its \emph{higher $\mu$-envelope} $\henvv{L}{\mu}$ is the set of all words in $L$ of maximal weight.
\begin{lem}[Envelopes] \label{lem:envelopes} 
Let $\Sigma$ be an alphabet, $L$ a finite language and $L'$ a regular language. Then the following hold:
    \begin{enumerate}[topsep=5pt,itemsep=2pt]
        \item[\textup{(a)}] $\rpn{L} \geq \rpn{ \henv{L} }$ and $\rpn{L'} \geq \rpn{\lenv{L'} }$.
        \item[\textup{(b)}]  $\rpn{L}\geq \rpn{\henvv{L}{\mu} }$ and $\rpn{L'}\geq \rpn{ \lenvv{L'}{\!\!\mu} }$ for every weighting $\mu : \Sigma \to \RRpos$.
    \end{enumerate}
\end{lem} 
\begin{proof}
(a)  We use the same construction as in \cite[Thm.\:2.4]{jerrum1982}; see also \cite[Clm.\:10]{jukna2015}.  To prove the first part, we transform an expression $R$ for the language $L$ into an expression for $\henv{L}$ without increasing its size. 
Traverse the syntax tree of $R$ in postorder to homogenize the described language of every node as follows. For a leaf there is nothing to do, so consider an inner node $u$ with children $x$ and $y$. By induction, the children $x,y$ already describe homogeneous languages. 
If $u=x+y$ is a union node and  $x$ and $y$ are of different degree, delete the child (and its whole subtree) of smaller degree.
If $u=x\cdot y$ is a concatenation node, there is nothing to do since concatenation preserves homogeneity.

After performing this procedure for all nodes, the resulting expression describes a homogeneous language. 
Since no word of $\henv{L}$ is deleted (we always keep the subtree with the words of larger length), the resulting language must be $\henv{L}$.

To show the second part of (a), first replace all star expressions by the empty word $\epsilon$, i.e., expressions of the form $T^*$ are replaced by $\epsilon$. 
Then proceed as above, but this time delete the children of \emph{larger} degree.

(b) The proof is analogous to (a), this time deleting the children whose described words have smaller resp.\ larger \emph{weight} (rather than their length). Further, we do \emph{not} replace star expressions $T^*$ that describe only words of zero weight.
\end{proof}

The next two corollaries give example applications of \cref{lem:envelopes}. 
The \emph{threshold language} $\thresh{n}{k} = \{ w \in \{0,1\}^n : |w|_1\geq k \}$  
is the monotone closure of the binomial language $\binomial{n}{k}$.
Mousavi \cite{mousavi2017} showed an optimal lower bound of $ \Omega(n \log^k (n))$ for $k\leq 3$.
We give a bound which is asymptotically optimal for $k=n^{\Theta(1)}$.
\begin{cor}[Threshold language] \label{cor:threshold}
Let $k\leq n/2$. Then the threshold language $\thresh{n}{k}$ requires regular expressions of length $\rpn{\thresh{n}{k}} \geq n  k^{\Omega( \log k) }.$
\end{cor}
\begin{proof}
Let $L=\thresh{n}{k}$ and define the weighting $\mu : \{0,1\} \to \RRpos$ as $\mu(0)=0, \mu(1)=1$. Then the weight of a word equals the number of its ones. The words of smallest weight are those with the fewest ones, thus $\lenvv{L}{\mu} = \binomial{n}{k}$. From \cref{lem:envelopes}\,(b) the inequality $\rpn{L}\geq \rpn{\binomial{n}{k}}$ follows, and \cref{cor:binomial} yields the lower bound $\rpn{\binomial{n}{k}} \geq n k^{\Omega(\log k)}$.
\end{proof}

\begin{cor}[Factors, prefixes, suffixes, subwords]\label{cor:prefix}
Let $L$ be a homogeneous language and let~$L'$ be the set of all factors, all prefixes, all suffixes or all subwords of $L$.
Then $L'$ requires regular expressions of length $\rpn{ L' } \geq \rpn{L}$.
\end{cor}
\begin{proof}
The upper envelope of $L'$ is $L$, so the claim follows from \cref{lem:envelopes}\,(a).
\end{proof}
\goodbreak
In \cref{cor:prefix} the condition of homogeneity is crucial. 
To see this, let $\Sigma$ be an alphabet and take any language $L \subseteq \Sigma^{n}$ with $\rpn{L} \geq \omega(n)$. Consider the non-homogeneous language $K \isdef  L \cup \Sigma^{n+1}$. Since the lower envelope of $K$ is $L$, \cref{lem:envelopes}\,(a) implies $\rpn{K} \geq \rpn{L} \geq \omega(n)$.
Let $K'$ be the set of all prefixes of~$K$. Then $K' = \Sigma^{\leq n+1}=\Sigma^0 \cup \Sigma^1 \cup \dots \cup \Sigma^{n+1}$ holds and $\rpn{K'} = O(n)$ follows.

\subsection{Infinite languages} \label{sec:infinite}

So far, all results in this paper were obtained only for \emph{finite} languages. 
However,  also bounds for infinite languages can be derived from them. 
This is particularly useful for languages with \emph{small star height} or languages whose star height is hard to determine. 
The star height $h(R)$ of an expression $R$  is defined as $h(R)=0$, if $R$ is a letter or the empty word $\epsilon$, and
$h(R_1 \cdot R_2) = h(R_1 + R_2) = \max( h(R_1), h(R_2)), ~h(R^*)=h(R)+1$.
That is, the star height is the maximal \emph{nesting depth} of star operations.
The star height $h(L)$ of a regular language~$L$ is the minimal star height of an expression describing $L$.
Gruber and Holzer showed the following:
\begin{thm}[Gruber and Holzer {\cite[Thm.\:6]{gruber2008ops1}}] \label{thm:starheight}
Let $L$ be a regular language.
Then $L$ requires regular expressions of length at least $\rpn{L} \geq \Omega(2^{h(L)/3})$. 
\end{thm}

Consider an infinite variant of the binomial language, namely $\binomialinf{n}{k} \isdef  \{w \in \{0,1\}^* \!: |w|_1\geq  k, |w|_0 \geq n-k\}$, 
and an infinite variant of the permutation language, namely $\Lperminf{n} \isdef \{w \in [n]^* : |w|_a \geq 1 ~\text{for all } a \in [n] \hspace*{1pt} \}$.
For either of these languages we can construct an expression of star height $h=1$ if we take  any expression describing their finite variant ($\binomial{n}{k} $ resp. $\Lperm{n}$) and replace every letter $a$ by $a\cdot \Sigma^*$.
Thus, both languages have star height~\mbox{$h=1$}, and only trivial bounds follow from \cref{thm:starheight}.
In contrast, with the help of envelopes and our previous results we obtain the following lower bounds.
\begin{cor} \label{cor:infinite}
Let $k\leq n/2$. Then the language $\binomialinf{n}{k}$ requires regular  expressions of length at least $\rpn{\binomialinf{n}{k}} \geq nk^{\Omega( \log k) }$.
\end{cor}
\begin{proof}
Let $L=\binomialinf{n}{k}$. We have $\lenv{L} =\binomial{n}{k}$, \cref{lem:envelopes}\,(a) implies $\rpn{L} \geq \rpn{\binomial{n}{k}}$, and  \cref{cor:binomial} yields the claimed lower bound.
\end{proof}
\begin{cor}\label{cor:perm-star}
The language $\Lperminf{n}$ requires regular expressions of length at least $\rpn{\Lperminf{n}} \!\geq  \Omega(4^n n^{- (3+\log n)/4})$.
\end{cor}
\begin{proof}
Let $L= \Lperminf{n}$. We have $\lenv{L} = \Lperm{n}$, \cref{lem:envelopes}\,(a) implies $\rpn{L} \geq \rpn{\Lperm{n}}$, and  \cref{thm:perm} yields the claimed lower bound.
\end{proof}

\section{Conclusion} \label{sec:conc}

We developed two lower bound methods for the length of regular expressions for finite languages.
With the \emph{arithmetic bound} (\cref{thm:arith}) we reduced expression length of homogeneous languages $L\subseteq \{0,1\}^n$ to monotone arithmetic formula complexity. This method naturally refines the  methods from Ellul et al. \cite{ellul2004} and Gruber and Johannsen~\cite{gruber2008cc} who gave reductions to boolean resp. monotone boolean formula complexity. 
Already with this method we solved two open problems: show a lower bound for the binomial language $\binomial{n}{k} = \{w \in \{0,1\}^n : |w|_1 =k\}$ and determine the blow-up of intersection and shuffle on finite languages.

With the \emph{log-product bound} (\cref{thm:logprod}) we provide a more general method which works for homogeneous languages over arbitrary alphabets, explicitly utilizes non-commutativity (see \cref{sec:utilize}) and even holds for expressions extended by a squaring operation (see \cref{rem:square}).
As applications we proved lower bounds for the divisibility language $\Ldiv{n}{p}$ of all binary numbers with $n$ bits that are divisible by $p$ and the parity language $\Leven{n}{k}= \{w \in [k]^n: |w|_a \text{ is even for all }a \in [k] \}$. These bounds are tight apart from small polynomial factors. For the language $\Lperm{n}$ of all permutations over $[n]$ we presented an alternative proof of the lower bound recently shown by Molina Lovett and Shallit \cite{shallit2019}.

Finally, with the help of \emph{envelopes} we showed how to transfer some results to non-homogeneous and infinite languages (see \cref{sec:envelope}).

\subsection{Open problems} \label{sec:problems}

We obtained bounds for several languages of ``combinatorial type'' (like the binomial or permutation language) or ``modulo arithmetic type'' (like the divisibility or parity language), but no result for a ``string matching type'' language, for example, the language $\Lnw$ of all length $n$ words that do \emph{not} contain a given word $w$ as a factor. 
In fact, this problem is the finite variant of a question already posed in \cite[Open Problem~3]{ellul2004}. 
For the choice $w=1^k$ we conjecture that $\rpn{ \Lnw } \geq n^{\Omega(\log k)}$ can be proven.
\begin{problem}
Prove a lower bound for the language $  \Lnw  \isdef  \Sigma^n \setminus (\Sigma^*  \:\! w \;\! \Sigma^*)$ for some word $w \in \Sigma^*$.
\end{problem}

The lower bound of $n k^{\Omega(\log k)}$ shown for the binomial language in \cref{cor:binomial} is  (asympotically) tight only for large values of $k$.
It would be nice to tighten this bound also for small values, like $k=O(\log n)$. Applying \cref{thm:logprod} together with an appropriate \emph{weighting} on the language (see \cref{rem:weightings}) might be a promising attempt.
\begin{problem}
Prove or disprove a lower bound $\rpn{\binomial{n}{k}} \geq n^{\Omega(\log k)}$ on the binomial language.
\end{problem}

Both our methods rely crucially on the \emph{homogeneity} of the languages. 
Envelopes can provide a loophole (see \cref{lem:envelopes}), but only in some special situations. 
For example, Jukna \cite{jukna2016} developed lower bound methods for non-homogeneous problems for tropical $(\max,+)$ circuits. It might be possible to transfer some ideas to regular expressions.
\begin{problem} 
Show lower bounds for finite, but non-homogeneous languages beyond the use of envelopes.
\end{problem}

Finally, possibly the most interesting question is whether methods from circuit complexity also work for \emph{infinite} languages, that is, for regular expressions \emph{with} star operations.
\begin{problem}
Does circuit complexity help proving lower bounds for regular expressions with star operations?
\end{problem}

\section*{Acknowledgments}
We wish to thank Mario Holldack, Stasys Jukna and Georg Schnitger for inspiring discussions and the anonymous referees for their kind and useful comments.

\bibliographystyle{elsarticle-num-names}
\bibliography{regex}{}

\begin{thebibliography}{38}
\expandafter\ifx\csname natexlab\endcsname\relax\def\natexlab#1{#1}\fi
\providecommand{\url}[1]{\texttt{#1}}
\providecommand{\href}[2]{#2}
\providecommand{\path}[1]{#1}
\providecommand{\DOIprefix}{doi:}
\providecommand{\ArXivprefix}{arXiv:}
\providecommand{\URLprefix}{URL: }
\providecommand{\Pubmedprefix}{pmid:}
\providecommand{\doi}[1]{\href{http://dx.doi.org/#1}{\path{#1}}}
\providecommand{\Pubmed}[1]{\href{pmid:#1}{\path{#1}}}
\providecommand{\bibinfo}[2]{#2}
\ifx\xfnm\relax \def\xfnm[#1]{\unskip,\space#1}\fi
\bibitem[{Cseresnyes and Seiwert(2020)}]{CS20}
\bibinfo{author}{E.~Cseresnyes}, \bibinfo{author}{H.~Seiwert},
\newblock \bibinfo{title}{{R}egular {E}xpression {L}ength via {A}rithmetic
  {F}ormula {C}omplexity},
\newblock in: \bibinfo{booktitle}{Descriptional Complexity of Formal Systems},
  \bibinfo{year}{2020}, pp. \bibinfo{pages}{26--38}.
  \DOIprefix\doi{10.1007/978-3-030-62536-8_3}.
\bibitem[{Ehrenfeucht and Zeiger(1976)}]{zeiger76}
\bibinfo{author}{A.~Ehrenfeucht}, \bibinfo{author}{P.~Zeiger},
\newblock \bibinfo{title}{Complexity measures for regular expressions},
\newblock \bibinfo{journal}{J. Comput. Syst. Sci.} \bibinfo{volume}{12}
  (\bibinfo{year}{1976}) \bibinfo{pages}{134--146}.
  \DOIprefix\doi{10.1016/S0022-0000(76)80034-7}.
\bibitem[{Ellul et~al.(2004)Ellul, Krawetz, Shallit, and Wang}]{ellul2004}
\bibinfo{author}{K.~Ellul}, \bibinfo{author}{B.~Krawetz},
  \bibinfo{author}{J.~Shallit}, \bibinfo{author}{M.-w. Wang},
\newblock \bibinfo{title}{Regular expressions: New results and open problems},
\newblock \bibinfo{journal}{J. Autom. Lang. Comb.} \bibinfo{volume}{9}
  (\bibinfo{year}{2004}) \bibinfo{pages}{233--256}.
  \DOIprefix\doi{10.25596/jalc-2004-233}.
\bibitem[{Gruber and Johannsen(2008)}]{gruber2008cc}
\bibinfo{author}{H.~Gruber}, \bibinfo{author}{J.~Johannsen},
\newblock \bibinfo{title}{Optimal lower bounds on regular expression size using
  communication complexity},
\newblock in: \bibinfo{booktitle}{Foundations of Software Science and
  Computational Structures (FoSSaCS)}, \bibinfo{year}{2008}, pp.
  \bibinfo{pages}{273--286}. \DOIprefix\doi{10.1007/978-3-540-78499-9\_20}.
\bibitem[{Gruber and Holzer(2008)}]{gruber2008ops1}
\bibinfo{author}{H.~Gruber}, \bibinfo{author}{M.~Holzer},
\newblock \bibinfo{title}{Finite automata, digraph connectivity, and regular
  expression size},
\newblock in: \bibinfo{booktitle}{35th International Colloquium on Automata,
  Languages, and Programming {(ICALP)}}, \bibinfo{year}{2008}, pp.
  \bibinfo{pages}{39--50}. \DOIprefix\doi{10.1007/978-3-540-70583-3\_4}.
\bibitem[{Gelade and Neven(2012)}]{gelade12}
\bibinfo{author}{W.~Gelade}, \bibinfo{author}{F.~Neven},
\newblock \bibinfo{title}{Succinctness of the complement and intersection of
  regular expressions},
\newblock \bibinfo{journal}{{ACM} Transactions on Computational Logic (TOCL)}
  \bibinfo{volume}{13} (\bibinfo{year}{2012}) \bibinfo{pages}{4:1--19}.
  \DOIprefix\doi{10.1145/2071368.2071372}.
\bibitem[{Mousavi(2017)}]{mousavi2017}
\bibinfo{author}{H.~Mousavi},
\newblock \bibinfo{title}{Lower bounds on regular expression size},
\newblock \bibinfo{journal}{CoRR} \bibinfo{volume}{abs/1712.00811}
  (\bibinfo{year}{2017}). \href{http://arxiv.org/abs/1712.00811}{{\tt
  arXiv:1712.00811}}.
\bibitem[{{Molina Lovett} and Shallit(2019)}]{shallit2019}
\bibinfo{author}{A.~{Molina Lovett}}, \bibinfo{author}{J.~Shallit},
\newblock \bibinfo{title}{Optimal regular expressions for permutations},
\newblock in: \bibinfo{booktitle}{46th International Colloquium on Automata,
  Languages, and Programming {(ICALP)}}, \bibinfo{year}{2019}, pp.
  \bibinfo{pages}{121:1--12}. \DOIprefix\doi{10.4230/LIPIcs.ICALP.2019.121}.
\bibitem[{Jerrum and Snir(1982)}]{jerrum1982}
\bibinfo{author}{M.~Jerrum}, \bibinfo{author}{M.~Snir},
\newblock \bibinfo{title}{Some exact complexity results for straight-line
  computations over semirings},
\newblock \bibinfo{journal}{J. ACM} \bibinfo{volume}{29} (\bibinfo{year}{1982})
  \bibinfo{pages}{874--897}. \DOIprefix\doi{10.1145/322326.322341}.
\bibitem[{Hrubes et~al.(2010)Hrubes, Wigderson, and Yehudayoff}]{hrubes10}
\bibinfo{author}{P.~Hrubes}, \bibinfo{author}{A.~Wigderson},
  \bibinfo{author}{A.~Yehudayoff},
\newblock \bibinfo{title}{Non-commutative circuits and the sum-of-squares
  problem},
\newblock in: \bibinfo{booktitle}{STOC}, \bibinfo{year}{2010}, pp.
  \bibinfo{pages}{667--676}. \DOIprefix\doi{10.1145/1806689.1806781}.
\bibitem[{Hrube{\v{s}} and Yehudayoff(2013)}]{non-commuting}
\bibinfo{author}{P.~Hrube{\v{s}}}, \bibinfo{author}{A.~Yehudayoff},
\newblock \bibinfo{title}{Formulas are exponentially stronger than monotone
  circuits in non-commutative setting},
\newblock in: \bibinfo{booktitle}{Conference on Computational Complexity},
  \bibinfo{year}{2013}, pp. \bibinfo{pages}{10--14}.
  \DOIprefix\doi{10.1109/CCC.2013.11}.
\bibitem[{Filmus(2011)}]{filmus2011}
\bibinfo{author}{Y.~Filmus},
\newblock \bibinfo{title}{Lower bounds for context-free grammars},
\newblock \bibinfo{journal}{Inf. Process. Lett.} \bibinfo{volume}{111}
  (\bibinfo{year}{2011}) \bibinfo{pages}{895--898}.
  \DOIprefix\doi{10.1016/j.ipl.2011.06.006}.
\bibitem[{Hrube{\v{s}} and Yehudayoff(2011)}]{yehudayoff2011}
\bibinfo{author}{P.~Hrube{\v{s}}}, \bibinfo{author}{A.~Yehudayoff},
\newblock \bibinfo{title}{Homogeneous formulas and symmetric polynomials},
\newblock \bibinfo{journal}{Comput. Complex.} \bibinfo{volume}{20}
  (\bibinfo{year}{2011}) \bibinfo{pages}{559--578}.
  \DOIprefix\doi{10.1007/s00037-011-0007-3}.
\bibitem[{Birget(1992)}]{foolingset1}
\bibinfo{author}{J.-C. Birget},
\newblock \bibinfo{title}{Intersection and union of regular languages and state
  complexity},
\newblock \bibinfo{journal}{Inf. Process. Lett.} \bibinfo{volume}{43}
  (\bibinfo{year}{1992}) \bibinfo{pages}{185--190}.
  \DOIprefix\doi{10.1016/0020-0190(92)90198-5}.
\bibitem[{Glaister and Shallit(1996)}]{foolingset2}
\bibinfo{author}{I.~Glaister}, \bibinfo{author}{J.~Shallit},
\newblock \bibinfo{title}{A lower bound technique for the size of
  nondeterministic finite automata},
\newblock \bibinfo{journal}{Inf. Process. Lett.} \bibinfo{volume}{59}
  (\bibinfo{year}{1996}) \bibinfo{pages}{75--77}.
  \DOIprefix\doi{10.1016/0020-0190(96)00095-6}.
\bibitem[{Khrapchenko(1971)}]{khrapchenko1971}
\bibinfo{author}{V.~M. Khrapchenko},
\newblock \bibinfo{title}{Method of determining lower bounds for the complexity
  of {P}-schemes},
\newblock \bibinfo{journal}{Mathematical Notes of the Academy of Sciences of
  the USSR} \bibinfo{volume}{10} (\bibinfo{year}{1971})
  \bibinfo{pages}{474--479}. \DOIprefix\doi{10.1007/BF01747074}.
\bibitem[{Grigni and Sipser(1995)}]{grigni}
\bibinfo{author}{M.~Grigni}, \bibinfo{author}{M.~Sipser},
\newblock \bibinfo{title}{Monotone separation of logarithmic space from
  logarithmic depth},
\newblock \bibinfo{journal}{J. Comput. Syst. Sci.} \bibinfo{volume}{50}
  (\bibinfo{year}{1995}) \bibinfo{pages}{433--437}.
  \DOIprefix\doi{10.1006/jcss.1995.1033}.
\bibitem[{Geffert et~al.(2010)Geffert, Mereghetti, and Palano}]{italian2010}
\bibinfo{author}{V.~Geffert}, \bibinfo{author}{C.~Mereghetti},
  \bibinfo{author}{B.~Palano},
\newblock \bibinfo{title}{More concise representation of regular languages by
  automata and regular expressions},
\newblock \bibinfo{journal}{Inf. Comput.} \bibinfo{volume}{208}
  (\bibinfo{year}{2010}) \bibinfo{pages}{385--394}.
  \DOIprefix\doi{10.1016/j.ic.2010.01.002}.
\bibitem[{Jukna(2012)}]{jukna2012}
\bibinfo{author}{S.~Jukna}, \bibinfo{title}{Boolean Function Complexity:
  Advances and Frontiers}, volume~\bibinfo{volume}{27},
  \bibinfo{publisher}{Springer Science \& Business Media},
  \bibinfo{year}{2012}.
\bibitem[{Shpilka and Yehudayoff(2010)}]{amir10:survey}
\bibinfo{author}{A.~Shpilka}, \bibinfo{author}{A.~Yehudayoff},
\newblock \bibinfo{title}{Arithmetic circuits: a survey of recent results and
  open questions},
\newblock \bibinfo{journal}{Found. Trends Theor. Comput. Sci.}
  \bibinfo{volume}{5} (\bibinfo{year}{2010}) \bibinfo{pages}{207–388}.
  \DOIprefix\doi{10.1561/0400000039}.
\bibitem[{Saptharishi(2015)}]{survey:github}
\bibinfo{author}{R.~Saptharishi},
\newblock \bibinfo{title}{A survey of lower bounds in arithmetic circuit
  complexity},
\newblock \bibinfo{journal}{Github survey}  (\bibinfo{year}{2015}). \URLprefix
  \url{https://github.com/dasarpmar/lowerbounds-survey/}.
\bibitem[{Gruber and Holzer(2009)}]{gruber2009ops2}
\bibinfo{author}{H.~Gruber}, \bibinfo{author}{M.~Holzer},
\newblock \bibinfo{title}{Tight bounds on the descriptional complexity of
  regular expressions},
\newblock in: \bibinfo{booktitle}{Developments in Language Theory},
  \bibinfo{year}{2009}, pp. \bibinfo{pages}{276--287}.
  \DOIprefix\doi{10.1007/978-3-642-02737-6\_22}.
\bibitem[{Hopcroft et~al.(2001)Hopcroft, Motwani, and Ullman}]{hopcroft2001}
\bibinfo{author}{J.~E. Hopcroft}, \bibinfo{author}{R.~Motwani},
  \bibinfo{author}{J.~D. Ullman},
\newblock \bibinfo{title}{Introduction to automata theory, languages, and
  computation},
\newblock \bibinfo{journal}{ACM SIGACT News} \bibinfo{volume}{32}
  (\bibinfo{year}{2001}) \bibinfo{pages}{60--65}.
\bibitem[{Gruber and Holzer(2015)}]{gruber2015survey}
\bibinfo{author}{H.~Gruber}, \bibinfo{author}{M.~Holzer},
\newblock \bibinfo{title}{From finite automata to regular expressions and back
  -- a summary on descriptional complexity},
\newblock \bibinfo{journal}{Int. J. Found. Comput. Sci.} \bibinfo{volume}{26}
  (\bibinfo{year}{2015}) \bibinfo{pages}{1009--1040}.
  \DOIprefix\doi{10.1142/S0129054115400110}.
\bibitem[{Jukna(2016)}]{jukna2016}
\bibinfo{author}{S.~Jukna},
\newblock \bibinfo{title}{Tropical complexity, sidon sets, and dynamic
  programming},
\newblock \bibinfo{journal}{SIAM J. Discrete Math.} \bibinfo{volume}{30}
  (\bibinfo{year}{2016}) \bibinfo{pages}{2064--2085}.
  \DOIprefix\doi{10.1137/16M1064738}.
\bibitem[{Yehudayoff(2019)}]{VPvsVNP}
\bibinfo{author}{A.~Yehudayoff},
\newblock \bibinfo{title}{Separating monotone {VP} and {VNP}},
\newblock in: \bibinfo{booktitle}{STOC}, \bibinfo{year}{2019}, pp.
  \bibinfo{pages}{425--429}. \DOIprefix\doi{10.1145/3313276.3316311}.
\bibitem[{Jukna(2015)}]{jukna2015}
\bibinfo{author}{S.~Jukna},
\newblock \bibinfo{title}{Lower bounds for tropical circuits and dynamic
  programs},
\newblock \bibinfo{journal}{Theory Comput. Syst.} \bibinfo{volume}{57}
  (\bibinfo{year}{2015}) \bibinfo{pages}{160--194}.
  \DOIprefix\doi{10.1007/s00224-014-9574-4}.
\bibitem[{Jukna(2016)}]{jukna16count}
\bibinfo{author}{S.~Jukna},
\newblock \bibinfo{title}{Lower bounds for monotone counting circuits},
\newblock \bibinfo{journal}{Discrete Appl. Math.} \bibinfo{volume}{213}
  (\bibinfo{year}{2016}) \bibinfo{pages}{139--152}.
  \DOIprefix\doi{10.1016/j.dam.2016.04.024}.
\bibitem[{Nisan(1991)}]{nisan91}
\bibinfo{author}{N.~Nisan},
\newblock \bibinfo{title}{Lower bounds for non-commutative computation
  (extended abstract)},
\newblock in: \bibinfo{booktitle}{STOC}, \bibinfo{year}{1991}, pp.
  \bibinfo{pages}{410--418}. \DOIprefix\doi{10.1145/103418.103462}.
\bibitem[{Valiant(1984)}]{valiant84}
\bibinfo{author}{L.~G. Valiant},
\newblock \bibinfo{title}{Short monotone formulae for the majority function},
\newblock \bibinfo{journal}{J. Algorithms} \bibinfo{volume}{5}
  (\bibinfo{year}{1984}) \bibinfo{pages}{363--366}.
  \DOIprefix\doi{10.1016/0196-6774(84)90016-6}.
\bibitem[{Chen(2008)}]{chungfeller}
\bibinfo{author}{Y.-M. Chen},
\newblock \bibinfo{title}{{The Chung–Feller theorem revisited}},
\newblock \bibinfo{journal}{Discrete Math.} \bibinfo{volume}{308}
  (\bibinfo{year}{2008}) \bibinfo{pages}{1328--1329}.
  \DOIprefix\doi{10.1016/j.disc.2007.03.068}.
\bibitem[{Shamir and Snir(1980)}]{shamir}
\bibinfo{author}{E.~Shamir}, \bibinfo{author}{M.~Snir},
\newblock \bibinfo{title}{On the depth complexity of formulas},
\newblock \bibinfo{journal}{Math. Syst. Theory} \bibinfo{volume}{13}
  (\bibinfo{year}{1980}) \bibinfo{pages}{301--322}.
  \DOIprefix\doi{10.1007/BF01744302}.
\bibitem[{Tiwari and Tompa(1994)}]{tiwari}
\bibinfo{author}{P.~Tiwari}, \bibinfo{author}{M.~Tompa},
\newblock \bibinfo{title}{A direct version of {S}hamir and {S}nir's lower
  bounds on monotone circuit depth},
\newblock \bibinfo{journal}{Inf. Process. Lett.} \bibinfo{volume}{49}
  (\bibinfo{year}{1994}) \bibinfo{pages}{243--248}.
  \DOIprefix\doi{10.1016/0020-0190(94)90061-2}.
\bibitem[{Seiwert(2020)}]{seiwert20}
\bibinfo{author}{H.~Seiwert},
\newblock \bibinfo{title}{{O}perational {C}omplexity of {S}traight {L}ine
  {P}rograms for {R}egular {L}anguages},
\newblock in: \bibinfo{booktitle}{Descriptional Complexity of Formal Systems},
  \bibinfo{year}{2020}, pp. \bibinfo{pages}{180--192}.
  \DOIprefix\doi{10.1007/978-3-030-62536-8_15}.
\bibitem[{Meyer and Stockmeyer(1972)}]{meyer72}
\bibinfo{author}{A.~R. Meyer}, \bibinfo{author}{L.~J. Stockmeyer},
\newblock \bibinfo{title}{The equivalence problem for regular expressions with
  squaring requires exponential space},
\newblock in: \bibinfo{booktitle}{13th Annual Symposium on Switching and
  Automata Theory}, \bibinfo{year}{1972}, pp. \bibinfo{pages}{125--129}.
  \DOIprefix\doi{10.1109/SWAT.1972.29}.
\bibitem[{Holzer and Kutrib(2011)}]{holzer11}
\bibinfo{author}{M.~Holzer}, \bibinfo{author}{M.~Kutrib},
\newblock \bibinfo{title}{The complexity of regular(-like) expressions},
\newblock \bibinfo{journal}{Int. J. Found. Comput. Sci.} \bibinfo{volume}{22}
  (\bibinfo{year}{2011}) \bibinfo{pages}{1533--1548}.
  \DOIprefix\doi{10.1142/S0129054111008866}.
\bibitem[{Diaconis et~al.(1990)Diaconis, Graham, and Morrison}]{hypercube2}
\bibinfo{author}{P.~Diaconis}, \bibinfo{author}{R.~L. Graham},
  \bibinfo{author}{J.~A. Morrison},
\newblock \bibinfo{title}{Asymptotic analysis of a random walk on a hypercube
  with many dimensions},
\newblock \bibinfo{journal}{Random Structures \& Algorithms}
  \bibinfo{volume}{1} (\bibinfo{year}{1990}) \bibinfo{pages}{51--72}.
  \DOIprefix\doi{10.1002/rsa.3240010105}.
\bibitem[{MacWilliams and Sloane(1977)}]{sloane}
\bibinfo{author}{F.~J. MacWilliams}, \bibinfo{author}{N.~J.~A. Sloane},
  \bibinfo{title}{The theory of error correcting codes},
  volume~\bibinfo{volume}{16}, \bibinfo{publisher}{Elsevier},
  \bibinfo{year}{1977}.

\end{thebibliography}

\appendix

\goodbreak
\section{Semirings and circuits} \label{sec:semirings} 

In this section we give some background for the connection between regular expressions without star and monotone arithmetic circuit complexity.
Although the concepts given here are not actually needed in our proofs, it might help the reader's intuition.

A \emph{semiring} $(S,\oplus,\otimes)$ consists of a set $S$ closed under
two binary associative   operations ``addition'' $(\oplus)$ and ``multiplication'' $(\otimes)$, where addition is commutative and multiplication
distributes over addition: $x\otimes(y\oplus z)=(x\otimes y)\oplus (x\otimes
z)$ and $(y\oplus z)\otimes x=(y\otimes x)\oplus (z\otimes x)$.\footnote{For convenience, in this section we understand ``$=$'' semantically rather than syntactically. For example, we say that the union operation ($+)$ is commutative since the expressions $(R_1+ R_2)$ and $(R_2 + R_1)$ describe the same languages although they do not coincide syntactically.}
Furthermore, the set $S$ contains an additive identity element $\nulll$ with $x \oplus \nulll =x$ and a multiplicative identity element $\eins$ with $\eins\otimes x=x\otimes\eins=x$.
A semiring is \emph{commutative} if its multiplication is commutative (that is, $ x\otimes y= y \otimes x$ holds), and is \emph{idempotent} if its addition is idempotent (that is, $x \oplus x = x$ holds).
We consider only semirings with characteristic zero (that is, $\eins \oplus \eins \oplus \dots \oplus \eins \neq \nulll$ holds for every sum of multiplicative identity elements) and with an absorbing additive identity element (that is, $\nulll \otimes  x = x \otimes \nulll = \nulll$ holds for all elements $x$).

Given a number $n \in \NN$, a \emph{circuit} over a semiring is a directed acyclic graph with one sink;
parallel edges joining the same pair of nodes are allowed.  
Each indegree-zero node (a \emph{leaf}) holds one of the
symbols $x_1,\ldots,x_n$ which are interpreted as formal \emph{variables} over $S$, or a semiring element interpreted as a \emph{constant}; we assume here that only the identity elements $\nulll$ and $\eins$ are allowed as constants.
Every other node (a \emph{gate}) has indegree two and performs one of the semiring
operations.  If the semiring is non-commutative, we assume that predecessors of gates are ordered, i.e., every gate has a left and a right predecessor.
 The sink is designated as the output gate.
The \emph{size} of a circuit is the total number of its nodes.  
A circuit whose each node has out-degree $\leq 1$ is a \emph{formula}; that is, formulas are circuits whose underlying graphs are \emph{trees}.

Since in any semiring $(S,\oplus,\otimes)$ multiplication distributes
over addition, each circuit over~$S$ defines (at the output gate) some \emph{formal polynomial}
\begin{align}\label{eq:poly-nc}
 f(x_1,\ldots,x_n)=\bigoplus_{m \in M} \const{m} X \langle m \rangle \ \ \mbox{ with }\ \ X \langle m \rangle  \isdef \bigotimes_{i=1}^{|m|} {m_i}\,
\end{align}
over~$S$ in a natural way, where $M \subseteq \{x_1,\dots,x_n\}^*$ is the set of all \emph{monomials} $X \langle m \rangle$  represented as vectors (or, equivalently, as words) $m=(m_1, \dots ,m_{|m|})$, and $\const{m} \in S \setminus \{\nulll \} $ are (non-zero) coefficients.
If the semiring is commutative, we can  write this polynomial also as  
\begin{equation}\label{eq:poly}
  f(x_1,\ldots,x_n)=\bigoplus_{a\in A}\const{a} X[a] \ \ \mbox{ with }\ \ X[a]  \isdef \bigotimes_{i=1}^n x_i^{a_i}\,
\end{equation}
where $A\subseteq \NN^n$ is its set of \emph{exponent vectors}, and $x_i^k$ stands for $x_i\otimes\cdots\otimes x_i$ $k$-times. An example is shown in \cref{fig:semiring}. 
Note the difference in the notation: In \cref{eq:poly-nc} the entry $m_i \in \{x_1, \dots, x_n\}$ denotes \emph{which variable} occurs in the monomial $X \langle m \rangle$ at position $i$, in \cref{eq:poly} the entry $a_i \in \NN$ denotes the \emph{power} (possible 0) of the $i$-th variable $x_i$ in the monomial $X[a]$.

In this paper, we mainly deal with two semirings: the \emph{free} and the \emph{arithmetic} semiring. 

\paragraph*{The free semiring}
Let $n \in \NN$, $\Sigma=\{x_1, \dots, x_n\}$ be an alphabet and $\reg{\Sigma}$ be the set of all regular expressions  over $\Sigma$ without stars. 
In the non-commutative, idempotent \emph{free} semiring $(\reg{\Sigma}, +, \Cdot)$  we
have $S=\reg{\Sigma}$, $x\oplus y  = x+y$ (union) and $x\otimes y  = x\cdot y$ (concatenation). 
The identity elements are $\nulll= \emptyset$ and $\eins= \epsilon$.  
The formal variables $x_i$ are interpreted as \emph{letters}. 
The polynomial given by \cref{eq:poly-nc} is 
$$f = \sum_{m \in M} \const{m} \:m_1 m_2 \cdots m_{|m|} \,.$$
Since the free semiring is idempotent ($x+x = x$ holds), all coefficients $\const{m}=\eins=\epsilon$ are trivial.
Thus, this polynomial $f$ corresponds to a regular expression $R_f$ (in ``sum-product normal form'' where all unions are at the top). 
In particular, any monomial $X \langle m \rangle = m_1 \otimes \cdots \otimes m_{|m|}$ is a word $m$ and the set $M$ of monomials is just the (finite) language $L(R_f)$ described by the expression $R_f$.
A \emph{formula} over the free semiring is nothing else than a regular expression without stars, viewed as its syntax tree.

\begin{rem}[Circuits for infinite regular languages] \label{rem:italian}
If we additionally allow unary gates performing the star operation, we obtain circuits that describe arbitrary regular languages, also called \emph{straight line programs} for regular languages. 
These circuits can simulate NFAs and can be exponentially more succinct.
For more information we refer to \cite{italian2010}; see also \cite{seiwert20} and references therein.
\end{rem}

\begin{figure}[t]
\begin{center}
\begin{tikzpicture}[semithick,yscale=0.52,xscale=0.95,scale=0.9]

\node[node] (1) at (3.4,0) { $\oplus$};
 
\node[node] (2) at (2,-2) { $\otimes$};  
\node[node] (3) at (4.8,-2) { $\oplus$}; 

\node[node] (7) at (1,-4) { $\oplus$}; 
\node[node] (5) at (3,-4) { $\otimes$}; 

\node[node] (11) at (0.55,-6) {$c$}; 
\node[node] (8) at (1.45,-6) {$a$}; 
\node[node] (8b) at (2.55,-6) {$b$}; 
\node[node] (9) at (3.45,-6) {$a$}; 
\node[node] (9b) at (4.35,-4) {$a$}; 
\node[node] (10) at (5.25,-4) {$a$}; 

\draw[->] (2)->(1);
\draw[->] (3)->(1);
\draw[->] (5)->(2);
\draw[->] (7)->(2);
\draw[->] (9)->(5);
\draw[->] (9)->(5);
\draw[->] (10)->(3);
\draw[->] (8b)->(5);
\draw[->] (9b)->(3);
\draw[->] (11)->(7);
\draw[->] (8)->(7);
\end{tikzpicture}
\vspace*{-2mm}
\end{center}
\caption{The  formula above computes the polynomial
$ (c  \otimes b \otimes a) \oplus (a \otimes b \otimes a) \oplus a \oplus a $; to improve readability we write $a$, $b$ and $c$ instead of $x_1$, $x_2$ resp. $x_3$.
If we plug in the free semiring, this formula defines 
the regular expression $cba+aba+a $ and describes the language $L=\{cba,aba,a\}$.
If we plug in the arithmetic semiring, the formula defines  the polynomial $abc+a^2b+2a$ and produces the set $A=\{(1,1,1),(2,1,0),(1,0,0)\}$.
}\label{fig:semiring}
\end{figure}
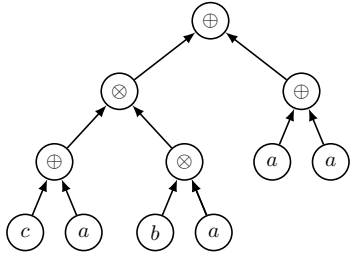

\paragraph*{The arithmetic semiring} \label{sec:arith-circs}

In the commutative \emph{arithmetic} semiring $(\RRpos ,+,\times)$, we have $S=\RRpos$, $x\oplus y  = x+y$ (addition) and $x\otimes y  = x \times y$ (multiplication). The identity elements are $\nulll =0$ and $\eins = 1$. Circuits over the arithmetic semiring are called monotone arithmetic circuits.
Here the polynomial given by \cref{eq:poly} defines a function over the non-negative reals
$$
f(x_1, \dots x_n)=\sum_{a\in A} \const{a} \prod_{i=1}^n x_i^{a_i} \,,
$$
with the set $A\subseteq \NN^n$ of exponent vectors and positive reals $\const{a}$ as coefficients. We call $A$ the \emph{produced set} of the circuit and, up to coefficients, the polynomial $f$ is determined by $A$.
The produced set $A \subseteq \NN^n$ of a circuit can also be defined recursively as follows:

A leaf holding the constant $0$ produces the empty set, a leaf holding the constant $1$ produces the set $\{\vec{0}\}$ and a leaf holding a variable~$x_i$ produces the set $\{\vec{e}_i\}$ where $\vec{e}_i$ is the unit vector with exactly one~$1$ at position~$i$ and~$0$ elsewhere.
An addition ($+$) gate produces the \emph{union} of the sets produced by its predecessors. 
Finally, a multiplication ($\times$) gate produces the \emph{Minkowski sum}~$\mplus$ of the sets produced by its predecessors; the Minkowski sum of two sets $A,B \subseteq \NN^n$ of vectors is $A \mplus B = \{a+b : a\in A,\, b\in B\}$. 
The set produced by the entire circuit is the set produced by its output gate. 
In \cref{table:semirings} we compare produced sets with described languages.

\begin{table}[t]
\begin{center}
\begin{tabular}{c|c|c}
~ & Free semiring & Arithmetic semiring   \\
\hline \\[-9pt]
Constants $\nulll$ and $\eins$  & $L(\emptyset)=\emptyset, \,L(\epsilon)=\{\epsilon\}$ &  $\prodd{0}=\emptyset, \,\prodd{1}=\{\vec{0}\}$ \\[3pt]
Variables/letters     $x_i$ & $L(x_i)=\{x_i\}$ & $\prodd{x_i}=\{\vec{e}_i\}$ \\[3pt]
Addition       $\oplus$& $L(R_1 + R_2) = L(R_1) \cup L(R_2)$ & $\prodd{u +v} = \prodd{u} \cup \prodd{v}$ \\[3pt]
Multiplication $\otimes$ & $L(R_1 \cdot R_2) = L(R_1)\cdot L(R_2) $  & $\prodd{u \times v} = \prodd{u}\mplus \prodd{v} $
\end{tabular}
\vspace*{-2mm}
\end{center}
\caption{Comparison between the free and the arithmetic semiring. Circuits over the free semiring describe languages $L$ while monotone arithmetic circuits produce sets $A$ of exponent vectors. Here, $\mplus$ denotes the Minkowski sum.}\label{table:semirings}
\end{table}

\paragraph*{The boolean semiring} \label{sec:boolean}
For completeness also the \emph{boolean semiring} $(\{0,1\},\lor, \land)$ should be mentioned.
Here we have $S=\{0,1\}$, $x \oplus y = x\lor y$ and $ x \otimes y = x \land y$, both operations are commutative and idempotent. The identity elements are $\nulll = 0$ and $\eins = 1$.
A circuit over the boolean semiring is a monotone boolean circuit.
The polynomial given by \cref{eq:poly} defines a monotone boolean function 
$$
f(x_1, \dots ,x_n)=\bigvee_{a\in A} \: \bigwedge_{i \in [n]: a_i \neq 0} x_i \;.
$$
Since also the ``multiplication'' of the boolean semiring is idempotent ($x\land x=x$ holds), often simplifications in the computed polynomials are possible. This is the reason why monotone boolean complexity can be smaller than monotone arithmetic complexity (see \cref{rem:arithvsbool}). 

\end{document}